\documentclass[a4paper,11pt]{article} 
\pdfoutput=1
\usepackage{amsmath,amsfonts,amsbsy,amsthm}
\usepackage{amssymb}
\usepackage{eucal}
\usepackage{hyperref}
\usepackage{cleveref}
\usepackage{mathrsfs}
\usepackage{comment}
\usepackage{dsfont}
\usepackage{verbatim} 
\usepackage{fourier}
\usepackage{authblk}

\usepackage{enumerate}
\usepackage{enumitem}
\setlist{leftmargin=*}
\usepackage{tikz}
\usetikzlibrary{decorations.pathmorphing,shapes,arrows,calc,decorations,decorations.markings,positioning}
\usepackage{pgfplots}
\usepackage{subfigure}

\numberwithin{equation}{section}

\makeatletter
\newtheoremstyle{corsivo}
   {\medskipamount}{\medskipamount}%
   {\itshape}{}%
   {\bfseries}{}%
   { }
   {\thmname{#1}\thmnumber{\@ifnotempty{#1}{ }\@upn{#2}}%
    \thmnote{ {\bfseries\boldmath(#3)}}.}%
\makeatother

\theoremstyle{corsivo}
\newtheorem{theorem}{Theorem}[section]

\newtheorem{lemma}[theorem]{Lemma}

\newtheorem{proposition}[theorem]{Proposition}
\newtheorem{definition}[theorem]{Definition}
\newtheorem{assumption}[theorem]{Assumption}

\makeatletter
\newtheoremstyle{dritto}
   {\medskipamount}{\medskipamount}%
   {\rmfamily}{}%
   {\bfseries}{}%
   { }
   {\thmname{#1}\thmnumber{\@ifnotempty{#1}{ }\@upn{#2}}%
    \thmnote{ {\bfseries\boldmath(#3)}}.}%
\makeatother

\theoremstyle{dritto}

\newtheorem{remark}[theorem]{Remark}

\newcommand{\sub}[1]{_{\mathrm{#1}}}

\newcommand{\eps}{\varepsilon}

\newcommand{\Id}{\mathds{1}}   

\newcommand{\di}{\mathrm{d}}

\newcommand{\X}{\mathcal{X}}

\newcommand{\N}{\mathbb{N}}
\newcommand{\Z}{\mathbb{Z}}

\newcommand{\R}{\mathbb{R}}
\newcommand{\C}{\mathbb{C}}

\newcommand{\Do}{\mathcal{D}}

\newcommand{\Hi}{\mathcal{H}}

\newcommand{\tuv}{\mathcal{T}}

\newcommand{\UB}{\mathcal{U}\sub{BF}}

\newcommand{\scal}[2]{\left\langle  #1 ,#2 \right\rangle}                
        
\newcommand{\norm}[1]{\left\| #1 \right\|}

\newcommand{\set}[1]{ \left\{  #1 \right\}} 

\DeclareMathOperator{\Tr}{Tr}         
\DeclareMathOperator{\tr}{tr}           
    
\DeclareMathOperator{\re}{Re}

\DeclareMathOperator{\dist}{dist}

\newcommand{\ie}{{\sl i.\,e.\ }}   
\newcommand{\eg}{{\sl e.\,g.\ }} 

\newcommand{\alev}{{\sl a.\,e.\ }} 

\newcommand{\Or}{{\mathrm{O}}}

\newcommand{\abs}[1]{\left\lvert#1\right\rvert}

\newcommand{\virg}[1]{``#1''}

\renewcommand{\(}{\left(}
\renewcommand{\)}{\right)}

\newcommand{\FC}{{\mathcal{C}}}

\newcommand{\fjj}{\Tilde{f}_{jj}}
\newcommand{\zjj}{\zeta_{jj}}
\newcommand{\Jj}{\Tilde{J_j}}
\newcommand{\fjjsing}{\fjj^{\text{sing}}}
\newcommand{\fjjreg}{\fjj^{\text{reg}}}

\newcommand{\Bril}{{\FC_1^*}}

\let\oldfootnote\footnote
\renewcommand{\footnote}[1]{\oldfootnote{\  #1}}

\title{
Longitudinal conductivity at integer quantum Hall transitions 
}
\author[1]{Giovanna Marcelli}
\author[2]{Lorenzo Pigozzi}
\author[3]{Marcello Porta}
\affil[1]{Dipartimento di Matematica e Fisica, Università di Roma Tre, Largo S. L. Murialdo 1, 00146
Roma, Italy}
\affil[2]{Institute of Science and Technology Austria, Am Campus 1, 3400 Klosterneuburg,
Austria}
\affil[3]{Mathematics Area, SISSA, Via Bonomea 265, 34136 Trieste, Italy}
\date{\today}

\begin{document}

\maketitle

\begin{abstract} 
We consider a class of two-dimensional tight binding models displaying conical intersections of the Bloch bands at the Fermi level. The setting includes the case of generic transitions between quantum Hall phases. We consider the longitudinal conductivity, as given by Kubo formula, describing the variation of the current after introducing a space-homogeneous electric field, in an adiabatic way. We obtain an explicit expression for the longitudinal conductivity, completely determined by the number of conical intersections and by the shape of the cones. In particular, the formula reproduces the known quantized values found for graphene and for the critical Haldane model. Furthermore, we discuss the validity of Kubo formula in presence of conical intersections in the spectrum, starting from the time-dependent Schr\"odinger equation. For electric fields which are weak and slowly varying in space and in time, we prove the validity of linear response from quantum dynamics.
\end{abstract}

\maketitle
\tableofcontents
\goodbreak

\section{Introduction}
\label{sec:intro} 
The Integer Quantum Hall Effect (IQHE) is a paradigmatic example of universal transport phenomenon in condensed matter physics. Experimentally observed in \cite{klitzing}, since then it has been the focus of extensive research in theoretical and in mathematical physics \cite{G}. The IQHE can be explained via the formalism of linear response theory: when a gapped Hamiltonian is adiabatically perturbed by an electric field of small intensity, the resulting transverse conductivity, defined via Kubo formula, is quantized in integer multiples of $e^{2} / h$, with $e$ the electric charge and $h$ Planck's constant. The first theoretical works on the IQHE date back to the seminal papers \cite{laughlin, tknn, avron}. Since then, mathematically rigorous results established the quantization of the Hall conductivity for general lattice models \cite{BES, AS2} and the emergence of Hall plateaux as a consequence of disorder and Anderson localization \cite{AG}. Concerning the validity of Kubo formula from quantum dynamics, it can be proved as a consequence of the adiabatic theorem \cite{ASY}. Actually, for the Hall effect, Kubo formula turns out to be exact, beyond linear response \cite{KS}; see also the recent generalization \cite{MM}. The validity of Kubo formula has been obtained for gapped models in the continuum \cite{ES,M}, and more recently for lattice systems in presence of strong disorder \cite{dREF}. Concerning gapless models, the validity of Kubo formula for one-dimensional systems and for the edge transport properties of quantum Hall systems has been derived in \cite{PS}, for non-interacting systems. In the last few years, the proof of quantization of the Hall conductivity has been extended to interacting fermions as well \cite{HM, GMPhall, BBdRF}. Also, the validity of Kubo formula has been justified in the many-body setting \cite{BdRF, BdRFrev, MT, henheik, GLMP24}, and its exactness for quantum Hall systems has been proved in \cite{BRFL, WMMMT}. Concerning interacting gapless one-dimensional systems, the validity of Kubo formula has been recently obtained in \cite{PSS}, combining the approach of \cite{GLMP24} with  rigorous renormalization group methods.

It is a well-known fact that two local, gapped Hamiltonians associated with different values of the Hall conductivity cannot be smoothly deformed one into the other, without closing the spectral gap at some point in the interpolation; this gap closure is referred to as a quantum Hall phase transition. Understanding the nature of this phase transition is an interesting problem in condensed matter physics. Let us consider the simplest case, that is non-interacting and translation-invariant systems. There, the spectral properties of the Hamiltonian can be analyzed in momentum space, via Bloch theory. In many concrete situations, it is found that at the transition point the spectral gap closes with conical intersections of the Bloch bands. For instance, this is the case for the Haldane model \cite{haldane}, or for the model considered in \cite{ludwig}. Recently, it has been proved that conical intersections at the Fermi level arise generically \cite{drouot} at Hall transitions, for a wide class of non-interacting, translation-invariant systems.

Models with conical intersections at the Fermi level are an example of semimetals, and display interesting transport properties. Let us consider the longitudinal conductivity, defined as the linear response of the average current to a parallel electric field. This quantity is exactly zero for gapped systems, but might be non-zero for gapless models. Specifically, the longitudinal conductivity has been computed in a number of semimetallic systems, such as graphene \cite{stauber}, the critical Haldane model \cite{haldane}, or in \cite{ludwig}. In these works, it is found that the longitudinal conductivity, defined via Kubo formula, is equal to $n / 16$, where $n$ is the number of conical intersections at the Fermi level. More recently, this remarkable formula has also been rigorously established for weakly interacting fermions, in the case of graphene \cite{GMP12} and for the critical Haldane-Hubbard model \cite{GMJP16, GMPhald}, using rigorous renormalization group methods and Ward identities. 

The purpose of this paper is twofold. In the first part, we compute the Kubo longitudinal conductivity for a general class of non-interacting lattice models, while in the second part we discuss the validity of Kubo formula for suitable time-dependent perturbations. Besides translation-invariance, our main assumption (spelled out more precisely in Assumption \ref{assum:H}) is the existence of conical intersections at the Fermi level for the Bloch bands, without any further hypothesis on the structure of the Bloch Hamiltonian. In particular, combined with \cite{drouot}, our theorem allows to determine the longitudinal conductivity at the generic transition between different Hall phases. In the second part, we discuss the validity of Kubo formula from quantum dynamics, for a smaller class of Hamiltonians displaying conical intersections, and for a suitable class of external fields which are weak and slowly varying in space-time. 

The first result of this paper can be written as follows.
\medskip

\noindent{\it Under Assumption \ref{assum:H}, the longitudinal conductivity in the direction $j=1,2$ is given by:}
    \begin{equation}
    \label{eqn:longitudinal-conductivity}
        \sigma_{jj} = \frac{1}{16}\sum_{l=1}^n\frac{s_{1j}^2 + s_{2j}^2}{|\det S_l|},
    \end{equation}
{\it where $n$ is the number of intersections between the Bloch bands at Fermi energy $\mu$ and $S_l = \big( s_{ij}\big)_{i,j=1,2}$,  is a $2\times 2$ invertible matrix that defines the shape of the cone near the $l$-th intersection, see  (\ref{eqn:conical}) (we omit the dependence on the cone label $l$ in $s_{ij}$ to simplify the notation). In particular, for isotropic cones, the right-hand side of (\ref{eqn:longitudinal-conductivity}) equals $\frac{n}{16}$.}
\medskip

Expression (\ref{eqn:longitudinal-conductivity}) depends solely on the eigenvalue structure of the Bloch Hamiltonian near the Fermi level and not on the Bloch functions. This is somewhat opposite to what happens for the IQHE, where the transverse conductivity is determined by the topology of the Bloch functions, and not by the energy bands. Thus, the longitudinal conductivity at criticality is identified as a non-topological transport property. In Remark \ref{rem:hald} we exhibit an explicit example of non-universality, by computing the longitudinal conductivity for the strained Haldane model.  The proof of (\ref{eqn:longitudinal-conductivity}) is based on an adaptation of the strategy introduced in \cite{GMP12} for interacting graphene, and then applied in \cite{GMJP16, GMPhald} to the critical Haldane-Hubbard model. Here we consider non-interacting fermions, however for a large class of Bloch Hamiltonians displaying conical intersections at the Fermi level, that allows to cover the generic case for quantum Hall transitions \cite{drouot}. With respect to the previous works \cite{GMP12, GMJP16}, our computation shows that the final expression only depends on the spectrum of a general $N\times N$ Bloch Hamiltonian close to the Fermi points, while the computations of \cite{GMJP16, GMPhald} used the explicit, asymptotic Dirac structure of the $2\times 2$ Bloch Hamiltonian close to the Fermi points.


A similar expression for the conductivity of two-dimensional systems exhibiting conical intersections of the Bloch bands at the Fermi level has been obtained in \cite{CFLSD}. The setting is however slightly different from the one of the present paper. In \cite{CFLSD}, the external field is introduced instantaneously, and the definition of conductivity involves a time-averaging assumption of the linear response. The assumptions of \cite{CFLSD} are more restrictive than ours (even though we believe that they could be relaxed): time-reversal symmetry is assumed, which rules out applications to the Hall effect, and stronger symmetry and regularity assumptions of the Bloch bands close to the conical intersections are required. The work \cite{CFLSD} considers isotropic cones, and obtains a universal result. We consider more general shapes of the conical crossing, and our formula shows under which circumstances the longitudinal conductivity takes a non-universal value.

Furthermore, our formulation of linear response can be rigorously justified from quantum dynamics, in a physically relevant regime. This is the content of the second part of the paper. We consider a smaller class of unperturbed Hamiltonians (we assume finite range hopping), on a finite lattice with side $L$ and with periodic boundary conditions, and we couple the system to an external vector potential $\theta A(\theta x) e^{\eta t}$, with $A(\cdot)$ the lattice restriction of a smooth and fast decaying function on $\mathbb{R}^{2}$. We assume that $\theta \in \mathbb{R}$ and $\eta \in \mathbb{R}_{+}$ are small, eventually to be taken to zero, but after the $L\to \infty$ limit. This scaling regime will facilitate the analysis of the real-time Duhamel expansion, which, due to the absence of a spectral gap for the unperturbed Hamiltonian, would othwerwise diverge at every order in the zero temperature and thermodynamic limit. Our choice of scaling is related to the one considered in  \cite{Frohlich}, where the derivation of effective actions for non-relativistic many-body systems is discussed from a non-rigorous viewpoint, for an external vector potential of the same type we consider but with $\theta = \eta$.

Given a local observable $O$, we denote by $\langle O\rangle_{\beta,L}$ its finite volume and finite temperature average at equilibrium, and by $\langle U(t;-\infty)^{*} O U(t;-\infty) \rangle_{\beta,L}$ its evolution under the dynamics generated by the time-dependent Hamiltonian. We prove the following result.

\medskip

\noindent{\it For $\eta, \theta$ small enough, such that $|\theta| \ll \eta / |\log \eta|$, and for $\beta, L$ large enough, we have:
\begin{equation}\label{eq:thm2inf}
\langle U(0;-\infty)^{*} O U(0;-\infty) \rangle_{\beta,L} - \langle O \rangle_{\beta,L} = (\langle U(0;-\infty)^{*} O U(0;-\infty) \rangle_{\beta,L} - \langle O \rangle_{\beta,L})^{\text{(1)}} + O(\theta^{2} \log \eta)\;,
\end{equation}
where the first term on the right-hand side is the linear response, as given by the first order of the Duhamel expansion in the vector potential for the perturbed dynamics, and where the error term is uniform in $\beta, L$.
}

We stress that the inverse temperature $\beta$ has to be taken larger than some inverse power of $\eta$, which prevents to use the result to study fixed positive temperatures as $\eta \to 0^{+}$. Notice that the full response of a physical observable $O$ must also take into account the coupling of $O$ with the electromagnetic field, defined via the Peierls' substitution. Choosing $O$ to be the current density operator, it is known that the first order variation of the state combined with the first order variation of the observable reproduces Kubo formula for the conductivity matrix, see {\it e.g.} \cite{GMP12, GMPhald, GMPhall}.  Observe that, for our choice of vector potential, the electric field at time zero is of order $\eta \theta$, which is much larger than the error terms in (\ref{eq:thm2inf}) provided $|\theta| \ll \eta / |\log\eta|$. Thus, (\ref{eq:thm2inf}) proves the validity of linear response in this range of parameters. It is useful to point out that, as the proof of Eq. (\ref{eq:thm2inf}) will show, the second order term in the Duhamel expansion becomes dimensionally comparable with the linear response for $\eta \sim \theta$, up to log-corrections. This does not take into account special cancellations in the perturbation theory, which might be present: in fact, for one-dimensional gapless systems, such cancellations indeed arise, as discussed in \cite{PS}, and more recently in \cite{PSS} for interacting systems. We defer the precise analysis of possible cancellations in the setting of the present paper to a future work. Let us also point out that the work \cite{CFLSD} discusses the validity of Kubo formula for gapless systems, for a certain range of parameters. There, it is shown that, for semimetallic systems, the averaged current over a period $T$ is well-approximated by linear response, provided the intensity $\varepsilon$ of the electric field is such that $\varepsilon T^{6} \ll 1$. Translated to our setting, $|\varepsilon| = \eta |\theta|$, and the time-scale $T$ plays the role of our $1/\eta$. Thus, the analysis of \cite{CFLSD} is relevant for $|\theta| \ll \eta^{5}$, which has to be compared with our range $|\theta| \ll \eta  / |\log \eta|$.

The proof of (\ref{eq:thm2inf}) is based on the analysis of the real-time Duhamel expansion, after rewriting each term as an integral over imaginary times. This representation has been proved in \cite{GLMP24} (see \cite{GMPhall, GMPweyl} for previous results at first and at second order), and its usefulness is that the oscillatory behavior of real-time correlations translates into decay properties of imaginary time correlations, which allow for a more efficient analysis of the Duhamel expansion. This identity has been used in \cite{GLMP24} to study transport in gapped many-body systems, and recently it has been the starting point of the study of transport in gapless systems as well; see \cite{PS} for the case of non-interacting $1d$ systems and the edge modes of $2d$ systems. As mentioned above, the analysis of \cite{PS} is based on the identification of certain cancellations at all orders in the Duhamel expansion, and on precise estimates of Euclidean correlations, which will not be needed in the present application, due to the higher co-dimension of the Fermi surface in the present work (it is $2$ here, and $1$ in \cite{PS}).

Concerning possible extensions of this paper, it would be interesting to study the critical transverse conductivity at the same level of generality of this work. See \cite{FGR} for recent work establishing the universality of the critical transverse conductivity for the Haldane model, that also takes into account the effect of many-body interactions. Furthermore, it would of course be very interesting to extend our results to the case of weakly interacting fermions. We expect this to be possible, using rigorous renormalization group methods; see \cite{GMP12, GMJP16, GMPhald, FGR} for the applications of these techniques to transport problems in two-dimensional interacting semimetallic systems. Finally, it would be interesting to extend the methods of the present paper to study spin transport, in the context of the quantum spin Hall effect; see {\it e.g.} \cite{MPTa, MPTe, MM2, FrMa} for recent rigorous results in this direction.

The paper is organized as follows. In Section \ref{sec:model} we define our setting and we state our results. In particular, in Section \ref{sec:reslin} we state Theorem \ref{thm:main}, about the longitudinal conductivity in presence of general conical intersections, while in Section \ref{sec:resder} we state Theorem \ref{thm:response}, about the derivation of linear response from quantum dynamics. Then, in Section \ref{sec:proof} we give the proof of Theorem \ref{thm:main}, while in Section \ref{sec:proofresp} we discuss the proof of Theorem \ref{thm:response}. Finally, in Appendix \ref{sec:aux} we collect some auxiliary technical results.

\section{Results} \label{sec:model}
\subsection{Linear response for quantum systems with conical intersections}\label{sec:linrisp}
\subsubsection{Setting}\label{sec:setting}
\paragraph{Periodic operators.} We consider infinite quantum systems having a \emph{crystalline structure}, namely their configuration space $\X$ is invariant under translations by any vector in a Bravais lattice $\Gamma$. We address  discrete models in $2$-dimension, namely $\X$ is a discrete set of points in $\R^2$. It can be assumed that the Bravais lattice $\Gamma$ is spanned over the integers by a basis $\set{v_1, v_2}\subset \R^2$.

The Hilbert space for a quantum particle is
$$
\Hi:=\ell^2(\X).
$$
The crystalline structure of the configuration space is lifted to a symmetry of the one-particle Hilbert space given by the unitary representation  of the \emph{translation operators} $T_\gamma$ with $\gamma\in\Gamma$, where
$$
\(T_\gamma\psi\)(x):=\psi(x-\gamma)\qquad \text{for every $\psi\in\Hi$}.
$$
An operator $A$ acting on $\Hi$ is called \emph{periodic} if $[A, T_\gamma] = 0$ for all $\gamma \in \Gamma$. As is well-known, the analysis of periodic operators is simplified by the use of their \emph{(modified)
Bloch--Floquet representation} (see \eg \cite[Chapter 5]{teufel} and references therein). At this point it is convenient to introduce some notation. 

We introduce the dual lattice $\Gamma^*:=\{ k\in\R^2\,:\, k\cdot\gamma\in 2\pi \Z \text{ for all $\gamma\in\Gamma$}  \}$. We denote by $\FC_1$ the \emph{centered fundamental cell of }  $\Gamma$, namely
\[
\FC_1:=\set{x\in\X\,:\, x=\sum_{j=1}^2\alpha_j v_j \text{ with $\alpha_j\in \left[-\frac{1}{2},\frac{1}{2}\right]$}}.
\]
Similarly, we define the \emph{unit fundamental cell of} $\Gamma^*$ by setting
\[
\FC_1^*:=\set{k\in\R^2\,:\, k=\sum_{j=1}^2\beta_j v^*_j \text{ with $\beta_j\in \left[-\frac{1}{2},\frac{1}{2}\right]$}},
\]
where $\{v_1^*, v_2^*\}$ is the dual basis of $\{v_1, v_2\}$, \ie $v_i^* \cdot v_j=2\pi \delta_{i,j}$.

The {\emph{(modified) Bloch--Floquet transform}} is initially defined on compactly supported functions $\psi\in \ell^2_c (\X)$ as
\begin{equation}
\label{eqn:UB}
(\UB\psi)(k,y) :=\frac{1}{\abs{\FC_1^*}^{1/2}}\sum_{\gamma \in \Gamma} e^{ -i k\cdot (y-\gamma)} (T_\gamma \psi)(y)\qquad \text{for all } k\in\R^2,\, y\in \X.
\end{equation}
Notice that for fixed $k\in \R^2$, the function $(\UB \psi)(k,\cdot)$ is periodic with respect to the translations by vectors in $\Gamma$, hence it can be considered as an element of $\ell^2(\X/\Gamma)\cong \C^N$, where $N<\infty$ is the cardinality of $\X/\Gamma$. On the other hand, for fixed $y\in\X$, the map $(\UB \psi)(\cdot,y)$ is \emph{pseudoperiodic} with respect to the translations by vectors in $\Gamma^*$:
\begin{equation}
\label{eqn:pseudoper}
(\UB \psi)(k+\gamma^*, y) = \left( \varrho_{\gamma^*} \, \UB \psi\right)(k , y)\qquad   \text{for all } \gamma^* \in \Gamma^*,
\end{equation}
where $(\varrho_{\gamma^*}\varphi)(y):=e^{-i \gamma^*\cdot y} \varphi(y)$ for every $\varphi\in\ell^2(\X/\Gamma)$, and $\Gamma^*\ni \gamma^* \mapsto\varrho_{\gamma^*}$ defines a unitary representation. 
It is useful to introduce the Hilbert space
\[
\Hi_\varrho:=\set{\phi\in L^2\sub{loc}\( \R^2, \ell^2(\X/\Gamma) \)\,:\, \phi(k+\gamma^*)=\varrho_{\gamma^*}\,\phi(k)\,\,\text{ for all $\gamma^*\in\Gamma^*$, for a.e. $k\in\R^2$}}
\]
equipped with the scalar product
\[
\scal{\phi}{\psi}_{\Hi_\varrho}:=\int_{\FC_1^*}dk\,\scal{\phi(k)}{\psi(k)}_{\ell^2(\X/\Gamma)}.
\]
The map defined by \eqref{eqn:UB} extends  to a unitary operator $\UB \colon \Hi \to \Hi_\varrho$. Its inverse transformation 
$\UB^{-1} \colon \Hi_\varrho\to\Hi$ is explicitly given by
$$
(\UB^{-1} \varphi)(x) =\frac{1}{\abs{\FC_1^*}^{1/2}}\int_{\FC_1^*}dk\,e^{i k\cdot x}\varphi(k,[x]),
$$
where $[\,\cdot\,]$ comes from to the a.e. unique decomposition $\X\ni x = \gamma_x+ [x]$, with $\gamma_x\in\Gamma$ and $[x] \in\FC_1$.
As anticipated before, this transform is advantageous in the analysis of periodic operators as they become
\emph{covariant fibered operator} acting in $\Hi_\varrho\subset L^2(\R^2, \ell^2(\X/\Gamma))=\int^{\oplus}_{\R^2}dk\, \ell^2(\X/\Gamma)$.
For a generic Hilbert space $\Hi_\sharp$ we denote by $\mathcal{L}(\Hi_\sharp)$ the set of all linear operators from $\Hi_\sharp$ to itself.
More specifically, given a periodic operator $A$ acting in $\Hi$, then one has that
\begin{equation} 
\label{eqn:A(k)}
\UB \, A \, \UB^{-1} = \int_{\R^2}^{\oplus}dk\, A(k),
\end{equation}
where each $A(k)$ is a element of $\mathcal{L}\(\ell^2(\X/\Gamma)\)\cong \mathcal{L}\(\C^N\)$, thus $A(k)$ is isomorphic to $N\times N$ matrix, and fulfills the \emph{covariance} property
\begin{equation}
\label{eqn:cov}
 A(k+\gamma^*) = \varrho_{\gamma^*} \, A(k) \, \varrho_{\gamma^*}^{-1} \qquad \text{for all } k\in\R^2, \; \gamma^*\in\Gamma^*. 
\end{equation}

\paragraph{The model.} Our goal is to investigate the response of a crystalline quantum system to the application of an external constant electric field of small intensity. Let us describe the system at equilibrium. Let $H$ be the Hamiltonian of the system, before the application of the electric field. The initial state of the system is defined by the \emph{Fermi projector} 
\begin{equation}
\label{eqn:fermiproj}
P_\mu:=\chi_{(-\infty,\mu]}(H)
\end{equation}
where $\mu$ is the \emph{Fermi energy} and $\chi_{(-\infty,\mu]}$ is the characteristic function of the set $(-\infty,\mu]$.

We shall consider Hamiltonians satisfying the following assumptions. In what follows, unless otherwise specified, we shall denote by $|k|$ the natural periodic norm of $k = (k_{1}, k_{2})$:
\begin{equation}
|k| = \min_{n_{1},n_{2} \in \mathbb{Z}} \big|k +  n_{1} v^*_1 +  n_{2} v^*_2\big|,
\end{equation}
where we recall that $(v^*_1, v^{*}_2)$ is the basis of the dual lattice $\Gamma^*$.
\begin{assumption} 
\label{assum:H}
\begin{enumerate}[label=$(\mathrm{H}_{\arabic*})$,ref=$(\mathrm{H}_{\arabic*})$]
\item \label{item:perselfadj}
The Hamiltonian $H$ of the unperturbed system is a periodic self-adjoint operator acting in $\Hi$ such that in Bloch--Floquet representation its fibration
\[
H\colon \R^2 \to \mathcal{L}(\ell^2(\X/\Gamma))\cong\mathcal{L}\(\C^N\),\quad k\mapsto H(k)
\]
is a $C^2$ covariant map.

\item \label{item:conical}
Let $\Lambda_+(k)$ be the smallest eigenvalue of $H(k)$ larger than or equal to $\mu$ and $\Lambda_-(k)$ is the largest eigenvalue of $H(k)$ smaller than or equal to $\mu$. The set $\left\{ k\in \FC_1^* \,:\, \mu \in \sigma(H(k)) \right\}$ consists of a finite number $n$ of distinct points $\omega_1,\dots,\omega_n$, and for each $\omega_l$ there exists an invertible $2\times 2$ matrix $S_l$ and a vector $a_l\in\R^2$:
\begin{equation}
\label{eqn:conical}
\Lambda_{\pm}(k)=\mu \pm \abs{S_l (k-\omega_l)}+a_l\cdot(k-\omega_l)+o\(\abs{k-\omega_l}\)\qquad \text{for $k\to \omega_l$.}
\end{equation}
In addition, we require that there exists a positive constant $C_*$ such that
\begin{equation}
\label{eqn:cone}
\abs{S_l(k)}-a_l\cdot k\geq C_*\abs{k}\qquad \text{for every $k\in\FC_1^*,\,1\leq l\leq n$}.
\end{equation}
Furthermore, we assume that for $k\neq \omega_{l}$ and for $|k-\omega_{l}|$ small enough, the multiplicity of $\Lambda_{\pm}(k)$ is $1$.
\end{enumerate}
\end{assumption}

With reference to Assumption \ref{item:conical}, the points $\omega_1,\dots,\omega_n$  are called the \emph{Fermi points}.

\begin{remark}
\label{rem:assum}

\begin{enumerate}[label=(\roman*), ref=(\roman*)]
\item \label{rem:1} The conical crossing assumption in the sense of equality \eqref{eqn:conical} is exactly the one singled out in \cite[Eq. (1.2)]{drouot}, relevant for generic transitions between Hall phases.

\item \label{rem:2} Observe that neglecting the remainder $ o\(\abs{k-\omega_l}\) $ equation \eqref{eqn:conical} describes a \emph{cone} in the classes of quadric surfaces under affine transformations. Indeed, let us recall that a generic affine transformation in $\R^3$ is given by 
\begin{equation}
(x_1,x_2,x_3)=x\mapsto x':=Ax+b
\end{equation}
where $A$ is an invertible $3\times 3$ matrix and $b$ is a generic vector in $\R^3$. The canonical equation of an affine cone is ${(x'_1)}^2+{(x'_2)}^2={(x'_3)}^2$.  We set
$$
A=\left(
\begin{array}{@{}c|c@{}}
S_l & \begin{array}{@{}c@{}} 0 \\ 0  \end{array} \\
\cline{1-1}
\multicolumn{1}{@{}c}{
  \begin{matrix} -a_l  \end{matrix}
} & 1
\end{array}
\right),\qquad
b=\begin{pmatrix}
0 \\
0 \\
-\mu 
\end{pmatrix},
$$
where $A$ is invertible since $S_l$ is so, and we use in $\R^3$ the coordinates $(k_1,k_2,z)$ while the new ones $(k'_1,k'_2,z')$ are obtained under the affine transformation induced by the above choices for $A$ and $b$. Clearly, ${(z')}^2={(k'_1)}^2+{(k'_2)}^2$ if and only if $z=\mu \pm \abs{S_l k}+a_l\cdot k$. A translation $(k_1,k_2)=k\mapsto k-\omega_l$ concludes this remark.
\item \label{rem:4} Let us show an explicit condition such that inequality \eqref{eqn:cone} is satisfied. Defining $\lambda_*:=\min\set{ \lambda_1,  \lambda_2\,:\,\text{ $\lambda_1$ and $\lambda_2$ are eigenvalues of $S_l^*S_l$ } }>0$. 
If $\sqrt{\lambda_*}-\abs{a_l}>0$ then inequality \eqref{eqn:cone} holds true. Indeed, being $V$ the unitary diagonalizing $S_l^*S_l$ we have that
\[
\abs{S_l k}^2=\scal{k}{S_l^*S_l k}=\scal{V k}{ V S_l^*S_l V^{-1} V k}\geq \lambda_*\abs{V k}^2=\lambda_*\abs{k}^2.
\]
Thus, by the Cauchy--Schwarz inequality we conclude that
\[
\abs{S_l k}-a_l\cdot k\geq (\sqrt{\lambda_*}-\abs{a_l})\abs{k}\qquad\text{for all $k\in\R^2$.}
\]

\item \label{it:spectraldec}
We denote by $\Lambda_1(k),\ldots,\Lambda_m(k)$ the eigenvalues of $H(k)$ smaller than $\mu$ \alev in $k$ and by $\Lambda_{m+1},\ldots,\Lambda_{N}(k)$ the eigenvalues bigger than $\mu$ \alev in $k$ (this labeling might includes repetition of the eigenvalues according to their multiplicity).
Notice that the number $m$ does not depend on $k$ due to Assumption \ref{assum:H}.

\item \label{it:eigenvalcont} 
The eigenvalues $\Lambda_1(k),\ldots,\Lambda_N(k)$ are continuous functions in $k$. 
In fact, they are roots of the characteristic polynomial of $H(k)$, denoted by $p_{H(k)}$.
The coefficients of $p_{H(k)}$ are $C^2$ in $k$ because $k\mapsto H(k)$ is $C^2$ according to hypothesis \ref{item:perselfadj}.
Therefore, the map $k \mapsto (\Lambda_1(k),\ldots,\Lambda_N(k))$ is continuous because the roots of a polynomial are continuous functions of its coefficients, see {\it e.g.} \cite{rootspolynomial}.
Then also the map $k\mapsto \Lambda_m(k) - \Lambda_{m+1}(k)\equiv \Lambda_-(k) - \Lambda_+(k)$ is continuous, and this will be used in Subsection \ref{sec:reular-part-fjj}.
\end{enumerate}
\end{remark}

\subsubsection{Linear response}\label{eq:linresp}

We couple the system to a spatially uniform electric field $E\in \R^2$, switched on adiabatically in time, via the Peierls' substitution. Denoting the adiabatic parameter with $0<\eta< 1$, we define the vector potential $A(t)$, for all $t\leq 0$:
\begin{equation}
\label{eqn:defnA(t)}
A(t):=-\int_{-\infty}^t ds\,E e^{\eta s}=-\frac{e^{\eta t}}{\eta} E.
\end{equation}
Next, we define the position operator in the $j$-th direction with $1\leq j\leq 2$ as 
\[
(X_j\psi)(x):=x_j \psi(x)\qquad\text{for all $\psi\in\Do(X_j)$}
\]
where $\Do(X_j)$ is its maximal domain. Setting $X:=(X_1,X_2)$, we introduce the gauge transformation
\begin{equation}
\label{eqn:gauge}
G(t) := e^{i A(t)\cdot X},
\end{equation}
and we define the time-dependent perturbed Hamiltonian $H(t)$ for all $t\leq 0$ as
\begin{equation}
\label{eqn:peierls-substitution}
H(t) = G(t) H G(t)^*.
\end{equation}
Clearly, $H(-\infty)\equiv H$. 
The Hamiltonian $H(t)$ is translation-invariant. By the standard properties of the Bloch--Floquet transform, the fibered Hamiltonian is:
\begin{equation}
\label{eqn:peierls-bloch}
H(t,k)= H(k-A(t)).
\end{equation}

Let us now discuss the evolution of the system, driven by the time-dependent electric field. The dynamics of the state is defined as the solution of the evolution equation: 
\begin{equation}
\label{eqn:pert-schrodinger}
\begin{cases}
 i\frac{\di}{\di t } \rho(t)  = [H(t), \rho(t)], \quad t\leq 0 \\
 \rho(-\infty) = P_\mu,
\end{cases}
\end{equation}
where $P_\mu$ is introduced in \eqref{eqn:fermiproj}. We will be interested in the response of the current operator, whose $j$-th component is:
\[
J_j(t):=i [H(t),X_j];
\]
its Bloch-Floquet fibration reads:
\begin{equation}
\label{eqn:currentk}    
J_j(t,k) = \partial_{k_j}H(t,k).
\end{equation}
For every $1\leq j,l\leq 2$, we shall introduce the conductivity matrix $\sigma_{jl}$ as the linear response coefficient of the current operator $J_j(t)$, defined via Kubo formula \cite{kubo}.

Let us recall the notion of trace per unit area, $\mathcal{T}(\,\cdot\,)$. For any $L\in 2\N+1$ we define:
\[
\FC_L:=\set{x\in\X\,:\, x=\sum_{j=1}^2\alpha_j v_j \text{ with $\alpha_j\in \left[-\frac{L}{2},\frac{L}{2}\right]$}}.
\]
Let $\chi_{\FC_L}\equiv \chi_L$ be the characteristic function of the set $\FC_L$. We define, for any operator\footnote{Since we are dealing with a discrete model, the operator $\chi_L A \chi_L$ is automatically trace class.} $A$ acting in $\Hi$:
\begin{equation}
\label{eqn:tuv}    
\tuv(A):=\lim_{\substack{L\to\infty\\ L\in 2\N +1}}\frac{1}{\abs{\FC_L}}\Tr\( \chi_L A \chi_L\),
\end{equation}
whenever the above limit exists. 

We are interested in the first-order term in the electric field of the expectation value of the current $J_j(0)$ in the $j$-th direction in the state $\rho(0)$ at time $t=0$, \ie when the perturbation is fully switched on\footnote{We adopt the notation $\Or_\eta(\abs{E}^2)$ to denote a function which can be bounded in norm by a pre-factor $\abs{E}^2$ but not uniformly in the adiabatic parameter $\eta$.}:
\begin{align*}
\tuv\(J_j(0)\rho(0)\)=&\frac{1}{(2\pi)^2}\int_{\FC_1^*}dk\,\Tr\( 
 J_j(0,k)\rho(0,k)\)\\
 =&\frac{1}{(2\pi)^2}\int_{\FC_1^*}dk\,\Tr\(\,J_j(k)P_\mu(k)\,\)+\frac{1}{\eta}\sum_{l=1}^2E_l\frac{1}{(2\pi)^2}\int_{\FC_1^*}dk\,\Tr(\,\partial^2_{jl}H(k) P_\mu(k)\,)\\
& -\frac{i }{\eta}\sum_{l=1}^2 E_l\frac{1}{(2\pi)^2}\int_{\FC_1^*}dk\,\int_{-\infty}^0\, dt\, e^{\eta t}\Tr\(\,J_j(k) \left[e^{iH(k)t}J_l(k)  e^{-iH(k)t}, P_\mu(k)\right] \,\)\\
&+O_\eta(\abs{E}^2),
\end{align*}
where first we have represented the trace per unit area as an integral in momentum space (see \eg \cite[Lemma 3]{panatisparberteufel}) and then we have expanded both the observable $J_j(0,k)$ and $\rho(0,k)$ in the electric field $E=(E_1,E_2)$.

It is convenient to introduce the momentum-space Heisenberg evolution as:
\begin{equation}
\label{eqn:interaction-transform}
\R\ni t\mapsto\tau_t(A(k)) := e^{iH(k)t}A(k)e^{-iH(k)t}.
\end{equation}
The above expansion motivates the following definition of the conductivity matrix.

\begin{definition}
\label{def:response coefficients}
    For all $1\leq j,l\leq 2$ the conductivity matrix $\sigma_{jl}$ 
    is defined as
    \begin{equation}
    \label{eqn:transport-coefficients}
        \sigma_{jl} := \lim_{\eta \to 0^+}\frac{1}{\eta}(f_{jl}(\eta) + \mathrm{s}_{jl})
    \end{equation}
    where
    \begin{align}
        \label{eqn:fjl-eta}
        f_{jl}(\eta) &:= \frac{i}{(2\pi)^2}\int_{\FC_1^*}dk\,\int_{-\infty}^0\, dt\, e^{\eta t}\Tr\(\,J_j(k) \left[P_\mu(k),\tau_t(J_l(k))\right] \,\) \\
        \label{eqn:schwinger-term}
        \mathrm{s}_{jl} &:= \frac{1}{(2\pi)^2}\int_{\FC_1^*} dk\Tr\left( \partial^2_{jl}H(k)P_\mu(k) \right).
    \end{align}
\end{definition}

The term $\mathrm{s}_{jl}$ is often called the Schwinger term, and later we will show how to express it in terms of $f_{jl}$.

\subsubsection{Evaluation of the linear response}\label{sec:reslin}

The next theorem gives an explicit expression for the longitudinal conductivity, that only depends on the conical structure of the energy bands at the Fermi level.

\begin{theorem}
\label{thm:main}
Under Assumption \ref{assum:H}, the longitudinal conductivity $\sigma_{jj}$ is given by, for $j=1,2$:
\begin{equation}
\label{eqn:main}
	\sigma_{jj} = \frac{1}{16}\sum_{l=1}^n \frac{s_{l,1j}^2 + s_{l,2j}^2}{|\det S_l|}
\end{equation}
where the scalar $s_{l,ij}$ is the $ij$-th element of the matrix $S_l$, appearing in the energy dispersion relation \eqref{eqn:conical} of the two eigenvalues $\Lambda_{\pm}$, near the Fermi energy $\mu$ around the Fermi point $\omega_l$.
\end{theorem}

Thus, the value of the longitudinal conductivity $\sigma_{jj}$ only depends on the shape of the conical intersection of the Bloch bands at the Fermi level. As our result shows, non-universality arises when the conical intersections are anisotropic. 
Moreover, observe that if for every Fermi point $\omega_l$ with $1\leq l\leq n$ the (real and invertible) matrix $S_l$ is orthogonal, up to the renormalization constant ${\abs{\det S}}^{-1/2}$, then each of the conical intersection contributes of $\frac{1}{16}$ to the quantized value of $\sigma_{jj}=\frac{n}{16}$. 
This is proven by the following elementary result, whose proof is postponed to Appendix \ref{sec:aux}.

\begin{lemma}
\label{lem:sufficient-longitudinal-quantization}
    Let $S$ be a $2\times 2$ real invertible matrix. Then we have that
    \begin{equation}
    \label{eqn:condS}
        \frac{s_{1j}^2 + s_{2j}^2}{|\det S|} = 1 \qquad \text{for every $j =1,2$}
    \end{equation}
    if and only if $s_{11}^2 + s_{21}^2 = s_{12}^2 + s_{22}^2$ and $(s_{11},s_{21})\cdot(s_{12}, s_{22})=0$, namely $S{\abs{\det S}}^{-1/2}$ is an orthogonal matrix. 
\end{lemma}

\begin{remark}\label{rem:hald}
\begin{enumerate}[label=(\roman*), ref=(\roman*)]

\item \label{rem:3} 
Theorem \ref{thm:main} covers the case of the Haldane model \cite{haldane} and its strained version \cite{mh} in the critical regime, where transitions between topological insulating phases occur (see \eg \cite{mamomopa} for a review without strain). These models have at most two Fermi points $\omega_1, \omega_2$.
If strain of intensity $\epsilon$ is applied in the armchair $y$ direction\footnote{I.e., one edge of the honeycomb structure is parallel to the $y$ axis and strain is applied along it.}, the matrix $S_l$ becomes diagonal, and the critical longitudinal conductivity is:
\begin{equation}
\sigma^{\mathrm{strain}}_{11} = \frac{1}{16}\frac{3+2\epsilon}{3-4\epsilon}, \qquad \sigma^{\mathrm{strain}}_{22} = \frac{1}{16}\frac{3-4\epsilon}{3+2\epsilon},
\end{equation}
where we have used \cite[Eqs. 14-17]{mh} for the form of the energy bands.
This shows that the longitudinal conductivity at the Hall transition is not a topological invariant, since away from criticality the topological phases of the Haldane model are independent of $\epsilon$; see \eg \cite[Figure 2]{mh}. See also \cite{strainrev} for a review about the effect of strain in the spectral and transport properties of graphene.

\item It would be interesting to prove that a similar breaking of universality is induced by many-body interactions, by using the renormalization group methods of \cite{GM, GMP12}. For instance, one could consider a non-interacting model with isotropic Dirac cones, perturbed by a many-body interaction that induces an anisotropic renormalization of the Fermi velocities. We plan to come back to this extension in a future work.
\end{enumerate}
\end{remark}

The proof of Theorem \ref{thm:main} is inspired by the analogous universality results obtained in \cite{GMP12, GMJP16, GMPhald}, for specific interacting models. A contribution of the present work is to show that, in a non-interacting setting, the longitudinal conductivity is insensitive to the form of the Bloch Hamiltonian, and it only depends on the conical intersections of the Bloch bands at the Fermi level. Furthermore, the result shows under which conditions the longitudinal conductivity deviates from quantization. 

The proof will be given in Section \ref{sec:proof}; let us briefly summarize the steps involved. We start by rewriting $\sigma_{jl}$ as the right derivative at $0$ of $(0,\infty)\ni\eta\to f_{jl}\in\R$, defined in \eqref{eqn:fjl-eta}. Then, we show that $f_{jj}$ is the restriction for positive $\eta$ of a suitable even function $\fjj\colon\R\to\R$, defined in \eqref{eqn:f_jj-tilde-definition}. The expression of this function is obtained from the original Duhamel integral, after a complex deformation of the integration domain, as in \eg \cite{GMPhall, GLMP24}. The $\eta$-dependence of $\tilde f_{jj}$ can be studied in a more efficient way. In particular, we prove that all the contributions to $\tilde f_{jj}(\eta)$ due to energies away from the Fermi level $\mu$ are smooth in $\eta$. Since $\tilde f_{jj}(\eta)$ is even in $\eta$, these contributions vanish as $\eta\to 0$. Finally, we are left with computing the contribution to $\sigma_{jj}$ associated with neighbourhoods of the conical intersections. As shown in Section \ref{sec:explicit-conductivity}, this contribution can be evaluated explicitly using the asymptotic form of the Bloch eigenvalues at the Fermi energy, recall \eqref{eqn:conical}.



\subsection{Kubo formula from quantum dynamics}


\subsubsection{Grand canonical formulation} 

\paragraph{Hamiltonian and Gibbs state.} Differently from Section \ref{sec:linrisp}, here we shall consider quantum systems on a finite lattice with periodic boundary conditions, and we will be interested in proving results that are uniform in the system's size. Also, we will consider the system at positive temperature, and we will take the zero temperature limit in a second moment. The setting of the previous section is recovered in the infinite volume limit, see \eg \cite[Appendix A]{FrMa}.

We shall represent the configuration space of the system $\mathcal{X}_{L}$ as $\mathcal{X}_{L} \cong \Lambda_{L} \times I_{N}$, where $\Lambda_{L} = \mathbb{Z}^{2} / L \mathbb{Z}^{2}$ is the lattice of the centers of the fundamental cells, and $I_{N} = \{1,\ldots, N\}$ collects the labels for the internal degrees of freedom (sublattice, spin, etc). We shall use the notation $\|x-y\|_{L}$ to denote the distance on the torus:
\begin{equation}
\|x-y\|_{L} := \min_{n_{1},n_{2} \in \mathbb{Z}} \| x-y + n_{1} L e_{1} + n_{2} L e_{2} \|\qquad \text{for any $x,y\in \Lambda_{L}$,}
\end{equation}
with $e_{1}, e_{2}$ the standard orthonormal basis of $\mathbb{R}^{2}$.

Let $H$ be a single-particle Hamiltonian on $\ell^{2}(\mathcal{X}_{L})$. We shall denote the integral kernel of this operator as $H_{\rho\rho'}(x,y)$ with $x,y\in \Lambda_{L}$ and $\rho,\rho'\in I_{N}$. As in Section \ref{sec:linrisp}, we assume translation invariance, $H_{\rho\rho'}(x,y) \equiv H_{\rho\rho'}(x-y)$. In addition to Section \ref{sec:linrisp}, we shall also assume that $H$ is finite-ranged, that is $H(x,y) = 0$ if $\|x-y\|_{L} > R$ for some $R$ independent of $L$. Observe that, since we are working on a finite lattice, the role of $\FC_{1}^{*}$ in Section \ref{sec:setting} is played by:
\begin{equation}\label{eq:CL}
\FC_{1,L}^{*} := \FC_{1}^* \cap \frac{2\pi}{L} \mathbb{Z}^{2}.
\end{equation}
We shall work in the second quantization formalism. The second quantization of $H$ is:
\begin{equation}
\mathcal{H} = \sum_{x,y \in \Lambda_{L}} \sum_{\rho,\rho'\in I_{N}} a^{*}_{x,\rho} H_{\rho\rho'}(x,y) a_{y,\rho'}\;,
\end{equation}
with $a_{x,\rho}, a^{*}_{x,\rho}$ the usual fermionic creation and annihilation operators, acting on the fermionic Fock space $\mathcal{F} = \mathbb{C}\oplus \bigoplus_{n\geq 1} \ell^{2}(\mathcal{X}_{L})^{\wedge n}$. We shall denote by $\tau_{t}(\cdot)$ the Heinsenberg dynamics generated by\footnote{Notice that this in a slight abuse of notation with respect to Section \ref{sec:linrisp}, where the same notation has been used to denote the Heisenberg evolution in first quantization.} $\mathcal{H}$:
\begin{equation}
\tau_{t}(\mathcal{A}) = e^{i\mathcal{H}t} \mathcal{A} e^{-i\mathcal{H}t}\;.
\end{equation}
In this second-quantized setting, the Peierls substitution for a general space-time dependent vector potential reads:
\begin{equation}\label{eq:peierls}
a^{*}_{x,\rho}a_{y,\rho'} \to a^{*}_{x,\rho}a_{y,\rho'} e^{i\int_{\ell_{\rho\rho'}(x,y)} d\ell \cdot A(t,\ell)}\;,
\end{equation}
where $\ell_{\rho\rho'}(x,y)$ is a straight oriented line connecting the point $x + r_{\rho}$ to the point $y + r_{\rho'}$, with $r_{\rho}, r_{\rho'}$ the relative coordinates of the particles within the fundamental cells labelled by $x$ and $y$. The time-dependent Hamiltonian is:
\begin{equation}
\mathcal{H}(t) = \sum_{x,y \in \Lambda_{L}} \sum_{\rho,\rho'\in I_{N}} a^{*}_{x,\rho} H^{A}_{\rho\rho'}(t; x,y) a_{y,\rho'}\;,
\end{equation}
with: 
\begin{equation}
H^{A}_{\rho\rho'}(t; x,y) := H_{\rho\rho'}(x,y) e^{i\int_{\ell_{\rho\rho'}(x,y)} d\ell \cdot A(t,\ell)}\;.
\end{equation}
We shall assume that the time-dependent vector potential $A$ has the form:
\begin{equation}
A(t,x) = A_{\theta}(x) e^{\eta t}\;,
\end{equation}
with $A_{\theta}(x) \equiv \theta A(\theta x)$ and where $\theta\in \mathbb{R}, \eta \in \mathbb{R}_{+}$ will eventually be sent to zero. Observe that the function $A(t,x)$ is defined for all $x\in \mathbb{R}^{2}$. The electric field generated by this vector potential is:
\begin{equation}\label{eq:Efield}
E(t,x) = - \theta \eta A(\theta x) e^{\eta t}\;.
\end{equation}
We shall suppose that $\|A(\cdot)\|_{L^{\infty}} \leq 1$, $\| A(0) \| = 1$, so that $\| E(0,0) \| = |\theta| \eta$. The parameter $\theta$ defines the amplitude and the space variation of the perturbation, while the parameter $\eta$ defines the time variation. In order to guarantee that the perturbed model is compatible with the periodic boundary conditions, we will suppose that $A(\theta x) = (A_{1}(\theta x), A_{2}(\theta x))$ is the periodization of a Schwartz function $A_{\infty}(\theta x)$ in $\mathbb{R}^{2}$:
\begin{equation}
A(\theta x) = \sum_{n_{1}, n_{2} \in \mathbb{Z}} A_{\infty}(\theta (x + n_{1} e_{1} L + n_{2} e_{2} L))\;.
\end{equation}
In Fourier space,
\begin{equation}\label{eq:hatA}
A_{\theta,\alpha}(x) = \frac{1}{L^{2}} \sum_{p\in \frac{2\pi}{L} \mathbb{Z}^{2}} e^{-ip\cdot x} \hat A_{\theta,\alpha}(p)\qquad \text{for all $x\in \mathbb{R}^{2} / L\mathbb{Z}^{2}$,}
\end{equation}
with $\hat A_{\theta,\alpha}(p) := (1/\theta)\hat A_{\infty,\alpha}(p/\theta)$. In particular, the following estimates hold:
\begin{equation}\label{eq:Aests}
\| \hat A_{\theta,\alpha}\|_{\ell^{\infty}}\leq \frac{C}{\theta}\;,\qquad \|\hat A_{\theta,\alpha}\|_{\ell^{1}} \leq C\theta\;.
\end{equation}
\paragraph{Dynamics in Fock space.} Let us denote by $\Gamma_{\beta,L}$ the grandcanonical Gibbs state of the system for $A=0$:
\begin{equation}
\Gamma_{\beta,L} = \frac{e^{-\beta (\mathcal{H} - \mu \mathcal{N})}}{\Tr_{\mathcal{F}} e^{-\beta (\mathcal{H} - \mu \mathcal{N})}}
\end{equation}
with $\mathcal{N}$ the number operator, $\mu$ the chemical potential and $\Tr_{\mathcal{F}}(\cdot)$ the trace over the fermionic Fock space. The evolution of the state is given by:
\begin{equation}\label{eq:dyn}
\begin{split}
i\partial_{t} \Gamma(t) &= [ \mathcal{H}(t), \Gamma(t) ]\;,\qquad t \leq 0\; \\
\Gamma(-\infty) &= \Gamma_{\beta,L}\;.
\end{split}
\end{equation}
For a given local observable $\mathcal{O}$, we will be interested in quantifying:
\begin{equation}
\Tr_{\mathcal{F}} \mathcal{O} \Gamma(0) - \Tr_{\mathcal{F}} \mathcal{O} \Gamma_{\beta,L}
\end{equation}
as $\beta, L\to \infty$ (the relative order will not be important), followed by the limits $\eta,\theta \to 0$. Our approach is based on a perturbative expansion. In order to set it up, let us rewrite the time-dependent Hamiltonian as:
\begin{equation}
\begin{split}
\mathcal{H}(t) &= \mathcal{H} + \sum_{n\geq 1} \frac{e^{n \eta t}}{n!} \mathcal{P}^{A}_{n} \\
&\equiv \mathcal{H} + \mathcal{P}(\eta t)\;,
\end{split}
\end{equation}
where the operator $\mathcal{P}^{A}_{n}$ is time-independent, and it is of order $n$ in $A$. This expression arises from the Taylor expansion of the exponential entering the definition of Peierls' substitution, Eq. (\ref{eq:peierls}). It is given by a sum of finite-ranged operators (the coupling with $A$ does not change the spatial range of the Hamiltonian). In momentum space, the operator $\mathcal{P}^{A}_{n}$ has the form:
\begin{equation}\label{eq:PnA}
\begin{split}
\mathcal{P}^{A}_{n} &= \frac{1}{L^{2}} \sum_{k\in \mathcal{C}_{1,L}^{*}} \frac{1}{L^{2n}} \sum_{p_{1}, \ldots, p_{n} \in \frac{2\pi}{L} \mathbb{Z}^{2}} \sum_{\alpha_{1}, \ldots, \alpha_{n} \in \{1,2\}} \sum_{\rho,\rho'\in I_{N}} \hat A_{\theta;\alpha_{1}}(p_{1}) \ldots \hat A_{\theta;\alpha_{n}}(p_{n}) \\&\qquad \cdot \hat a^{*}_{k + p_{1} + \ldots + p_{n},\rho} \hat a_{k,\rho'} \hat Q_{n;\underline{\alpha}, \rho, \rho'}(k, \underline{p})\;,
\end{split}
\end{equation} 
where $\hat a_{k}$ and $\hat a_{k}^{*}$ are the fermionic creation and annihilation operators in momentum space, 
\begin{equation}
\hat a_{k,\rho} := \sum_{x\in \Lambda_{L}} e^{ik\cdot x} a_{x,\rho}\;,\qquad \hat a^{*}_{k,\rho} := \sum_{x\in \Lambda_{L}} e^{-ik\cdot x} a^{*}_{x,\rho}\;,\qquad k \in \FC_{1,L}^{*},
\end{equation}
and $Q_{n;\underline{\alpha}, \rho, \rho'}(k, \underline{p})$ are suitable kernels, with the notation $\underline{p} = (p_{1}, \ldots, p_{n})$, $\underline{\alpha} = (\alpha_{1},\ldots, \alpha_{n})$. The explicit form of the kernels will not be needed; we will only use that, uniformly in all momenta:
\begin{equation}
|\hat Q_{n;\underline{\alpha}, \rho, \rho'}(k,\underline{p})|\leq C^{n}\;.
\end{equation}
The expansion (\ref{eq:PnA}) is obtained writing the argument of the exponent of the Peierls phase as, setting $x(\rho) = x + r_{\rho}$, $y(\rho') = y + r_{\rho'}$:
\begin{equation}\label{eq:intlin}
\begin{split}
\int_{\ell_{\rho\rho'}(x,y)} d\ell\cdot A_{\theta}(\ell) &= \int_{0}^{1} ds\, \frac{(x(\rho)-y(\rho'))}{\|x(\rho)-y(\rho')\|}\cdot \theta A(\theta (s x(\rho) + (1-s) y(\rho'))) \\
&= \sum_{\alpha=1,2} \frac{1}{L^{2}} \sum_{p\in \frac{2\pi}{L}\mathbb{Z^{2}} } \int_{0}^{1} ds\, \hat A_{\theta;\alpha}(p) \frac{(x(\rho)-y(\rho'))_{\alpha}}{\|x(\rho)-y(\rho')\|} e^{i p \ (s x(\rho) + (1-s)y(\rho'))} \\
&= \sum_{\alpha=1,2} \frac{1}{L^{2}} \sum_{p\in \frac{2\pi}{L}\mathbb{Z^{2}} } \hat A_{\theta;\alpha}(p) e^{ip\cdot y(\rho')} \frac{(x(\rho)-y(\rho'))_{\alpha}}{\|x(\rho)-y(\rho')\|}  \frac{e^{ip\cdot (x(\rho)-y(\rho'))} - 1}{ip\cdot (x(\rho)-y(\rho'))} \\
&\equiv \sum_{\alpha = 1,2} \frac{1}{L^{2}} \sum_{p\in \frac{2\pi}{L}\mathbb{Z^{2}} } \hat A_{\theta;\alpha}(p) e^{ip\cdot y} \eta^{\rho\rho'}_{\alpha,p}(x-y)\;.
\end{split}
\end{equation}
Observe that the left-hand side is trivially zero if $\|x(\rho) - y(\rho')\| = 0$. In this case, the function $\eta^{\rho\rho'}_{\alpha,p}(x-y)$ is defined as zero. Furthermore, $|\eta^{\rho\rho'}_{\alpha,p}(x-y)|\leq 1$. If $p\cdot (x(\rho) - y(\rho')) = 0$, the function $\eta^{\rho\rho'}_{\alpha,p}(x-y)$ is understood as being $e^{ip\cdot r_{\rho'}}(x(\rho) - y(\rho'))_{\alpha} / \| x(\rho) - y(\rho') \|$. Expanding the Peierls exponential in Taylor series, one obtains (\ref{eq:PnA}) after using the rewriting (\ref{eq:intlin}), and performing the lattice sums over the $x,y$ coordinates of the fermionic creation and annihilation operators.

\subsubsection{Validity of linear response}\label{sec:resder}

In the following, we shall restrict the attention to local observables $\mathcal{O} \equiv \mathcal{O}_{z}$ of the form:
\begin{equation}
\mathcal{O}_{z} = \sum_{\substack{x,y\in \Lambda_{L} \\ \{x,y\} \ni z}} \sum_{\rho, \rho'\in I_{N}} a^{*}_{x,\rho} O_{\rho,\rho'}(x,y) a_{y,\rho'}\;,
\end{equation}
with $O_{\rho,\rho'}(x,y) \equiv O_{\rho,\rho'}(x-y)$ and $O_{\rho,\rho'}(x,y)$ finite-ranged. This expression includes as special cases the density operator, or the current density operator associated with a finite-ranged Hamiltonian. More general local observables can be obtained summing over $z$ in appropriate subsets $X\subset \Lambda_{L}$. In what follows, we will denote by $\eta_{\beta}$ the best approximation of $\eta$ in $\frac{2\pi}{\beta} \mathbb{N}$, such that $\eta_{\beta} \geq \eta$. We will prove the following result.
\begin{theorem}\label{thm:response} For $\beta$ large enough, and for $|\theta| / \eta$ small enough uniformly in $\eta, \beta, L$:
\begin{equation}\label{eq:response}
\Tr_{\mathcal{F}} \mathcal{O}_{z} \Gamma(0) - \Tr_{\mathcal{F}} \mathcal{O}_{z} \Gamma_{\beta,L} = -i \int_{-\infty}^{0} dt\, e^{\eta_{\beta} t} \Tr_{\mathcal{F}} \Big[ \mathcal{O}_{z}, \tau_{t}\Big(\mathcal{P}^{A}_{1}\Big) \Big] \Gamma_{\beta,L} + R_{\mathcal{O}} 
\end{equation}
where 
\begin{equation}\label{eq:Rest}
\big| R_{\mathcal{O}} \big| \leq C_{1}\theta^{2} |\log \eta| + \frac{C_{2}}{\eta^{4} \beta}
\end{equation}
with constants $C_{1}, C_{2}$ independent of $\beta, L, \eta, \theta$.
\end{theorem}
\begin{remark}\label{rem:resp}
\begin{enumerate}
\item The first term in the right-hand side of (\ref{eq:response}) is the linear response of the observable $\mathcal{O}_{z}$. As the proof will show, the first order is bounded by $(\text{const.}) |\theta|$ uniformly in $L,\beta,\eta$.
Instead, the error term  $R_{\mathcal{O}}$ collects the error introduced by replacing the adiabatic parameter $\eta$ by $\eta_{\beta}$, and the terms of order higher than $1$ in $A$ appearing in the Duhamel expansion with  adiabatic parameter $\eta_{\beta}$.
%
%

\item Eq. (\ref{eq:Rest}) shows that the higher order terms vanish after dividing them by the intensity of the electric field $\theta \eta$, provided the limit $\beta \to \infty$ is taken before $\eta,\theta \to 0$, and provided:
\begin{equation}\label{eq:etatheta}
\frac{\theta}{\eta} |\log \eta| \to 0\;.
\end{equation}
Thus, the setting is relevant for describing the response of the system to approximately space-homogeneous electric fields (on length scale $\ll 1/\theta$). 

\item As mentioned in the introduction, the proof of Theorem \ref{thm:response} is based on the representation of the Duhamel expansion for the non-autonomous dynamics (\ref{eq:dyn}) in imaginary time, proven in \cite{GLMP24}. In \cite{GLMP24}, this representation has been used to study the response of gapped systems in presence of weak many-body interactions, using cluster expansion methods. Recently, this strategy has been used in \cite{PS} to study the response of non-interacting gapless systems as well, and also weakly interacting gapless systems \cite{PSS}. The advantage of the representation of perturbation theory introduced in \cite{GLMP24} is that imaginary-time correlations can be estimated in an efficient way. In particular, in certain scaling regimes, like the ones considered in \cite{PS, PSS} and in the current paper, the strategy allows to prove the convergence of the Duhamel series. We stress that the same result cannot be obtained via a direct application of standard adiabatic methods, since they crucially rely on the presence of a spectral gap. The proof of Theorem \ref{thm:response}  is based on estimates for imaginary-time correlation functions, which in the absence of interactions can be expressed in terms of products of the two-point function, using the fermionic Wick's rule. With respect to \cite{PS}, the proof of validity of linear response in the present setting turns out to be less involved, thanks to the higher co-dimension of the Fermi surface (which is $2$ in the present case, and $1$ in \cite{PS}), and to the chosen scaling of $\eta$ and $\theta$. 
\end{enumerate}
\end{remark}

Theorem \ref{thm:response} only describes the variation of the expectation value of the observable due to the change of the state. The full response of the system should also include the contribution due to the variation of the observable, after coupling it to the electromagnetic field via the Peierls' substitution (\ref{eq:peierls}). Let $O$ be the current operator in direction $j$. Then, $\mathcal{O}_{z} \equiv \mathcal{J}_{z,j}$ is the second quantization of the current density operator at $z$. Let $\mathcal{J}_{z,j}^{A}$ be the current coupled to the vector potential at time zero, as given by (\ref{eq:peierls}). The full variation of the average current is:
\begin{equation}
\Tr_{\mathcal{F}} \mathcal{J}^{A}_{z,j} \Gamma(0) - \Tr_{\mathcal{F}} \mathcal{J}_{z,j} \Gamma_{\beta,L}\;.
\end{equation}
We rewrite this difference as:
\begin{equation}\label{eq:jresp}
\begin{split}
&\Tr_{\mathcal{F}} \mathcal{J}^{A}_{z,j} \Gamma(0) - \Tr_{\mathcal{F}} \mathcal{J}_{z,j} \Gamma_{\beta,L} \\
&\quad= \Big(\Tr_{\mathcal{F}} \mathcal{J}_{z,j} \Gamma(0) - \Tr_{\mathcal{F}} \mathcal{J}_{z,j} \Gamma_{\beta,L}\Big) + \Big(\Tr_{\mathcal{F}} \mathcal{J}^{A}_{z,j} \Gamma(0) - \Tr_{\mathcal{F}} \mathcal{J}_{z,j} \Gamma(0)\Big)\;.
\end{split}
\end{equation}
The first term is studied using Theorem \ref{thm:response}. For the second term, we write:
\begin{equation}
\mathcal{J}^{A}_{z,j} - \mathcal{J}_{z,j} = \mathcal{J}^{A,(1)}_{z,j} + \mathcal{J}^{A,(\geq 2)}_{z,j}
\end{equation}
where $\mathcal{J}^{A,(1)}_{z,j}$ is the first order in the Taylor expansion in $A$ while $ \mathcal{J}^{A,(\geq 2)}_{z,j}$ are the higher order terms. It is easy to see that:
\begin{equation}\label{eq:tay}
\| \mathcal{J}^{A,(\geq 2)}_{z,j} \| \leq C \theta^{2}\;,
\end{equation}
for a constant $C$ that does not depend on $\beta,L,\eta,\theta$, and that depends on the range of the Hamiltonian. Therefore, we write the second term in (\ref{eq:jresp}) as, for $\beta$ large enough:
\begin{equation}
\begin{split}
\Big(\Tr_{\mathcal{F}} \mathcal{J}^{A}_{z,j} \Gamma(0) - \Tr_{\mathcal{F}} \mathcal{J}_{z,j} \Gamma(0)\Big)  &= \Tr_{\mathcal{F}} \mathcal{J}^{A,(1)}_{z,j} \Gamma(0) + R_{\mathcal{J},1} \\ 
& = \Tr_{\mathcal{F}} \mathcal{J}^{A,(1)}_{z,j} \Gamma_{\beta,L} + R_{\mathcal{J},1} + R_{\mathcal{J},2}\;,
\end{split}
\end{equation}
where the first step follows from the Taylor expansion of the Peierls' phase, while the second term follows from the application of Theorem \ref{thm:response}. Both error terms are bounded as  $| R_{\mathcal{J},i} |\leq C\theta^{2}$, by (\ref{eq:tay}) and by the fact that the linear response of the observable $\mathcal{J}^{A,(1)}_{z,j}$ is of order $\theta^{2}$.

Thus, using Theorem \ref{thm:response} for the first term in (\ref{eq:jresp}):
\begin{equation}\label{eq:Jresp}
\begin{split}
&\Tr_{\mathcal{F}} \mathcal{J}^{A}_{z,j} \Gamma(0) - \Tr_{\mathcal{F}} \mathcal{J}_{z,j} \Gamma_{\beta,L} \\
&\quad = -i \int_{-\infty}^{0} dt\, e^{\eta_{\beta} t} \Tr_{\mathcal{F}} \Big[ \mathcal{J}_{z,j}, \tau_{t}\Big(\mathcal{P}^{A}_{1}\Big) \Big] \Gamma_{\beta,L}  + \Tr_{\mathcal{F}} \mathcal{J}^{A,(1)}_{z,j} \Gamma_{\beta,L} + R_{\mathcal{J},1} + R_{\mathcal{J},2} + R_{\mathcal{J},3}
\end{split}
\end{equation}
where the last error term is bounded as in (\ref{eq:Rest}). The sum of the first two terms in (\ref{eq:Jresp}) turns out to be of order $\theta \eta$, unlike the two terms individually. Let us choose the vector potential $A$ to be directed in the $l$ direction. Let:
\begin{equation}\label{eq:Kjldef}
K_{jl}(\eta,\theta) := 
\lim_{\beta, L\to \infty} \frac{1}{\theta \eta} \Big(-i \int_{-\infty}^{0} dt\, e^{\eta_{\beta} t} \Tr_{\mathcal{F}} \Big[ \mathcal{J}_{0,j}, \tau_{t}\Big(\mathcal{P}^{A}_{1}\Big) \Big] \Gamma_{\beta,L}  + \Tr_{\mathcal{F}} \mathcal{J}^{A,(1)}_{0,j} \Gamma_{\beta,L} \Big)\;.
\end{equation}
For $\eta$ fixed, the order of the $\beta,L\to \infty$ limits is not important. Observe that we are evaluating the current at the point $z=0$, where $\| E(0,0) \| = |\theta| \eta$; thus, the prefactor in the right-hand side of (\ref{eq:Kjldef}) is the natural normalization needed to define the conductivity matrix. The function $K_{jl}(\eta,\theta)$ has been studied in many previous works devoted to the rigorous analysis of linear response for interacting gapless systems, see {\it e.g.} \cite{GMP12, GMPhall, GMPhald}. In particular, one can prove that:
\begin{equation}\label{eq:Klip}
\big|K_{jl}(\eta,\theta) - K_{jl}(\eta,\theta')\big|\leq \frac{C}{\eta}|\theta - \theta'|\;. 
\end{equation}
For non-interacting systems, this bound can also be proved proceeding as in Section \ref{sec:proofresp} of this work. Also, it is known that \cite{GMP12, GMPhall, GMPhald}:
\begin{equation}\label{eq:Klim}
\lim_{\eta \to 0^{+}} \lim_{\theta \to 0}K_{jl}(\eta,\theta) = \sigma_{jl}
\end{equation}
with $\sigma_{jl}$ as in Definition \ref{def:response coefficients}. By translation-invariance of the equilibrium Gibbs state, one can actually replace the evaluation at $z=0$ with the trace per unit volume (taken after the $\beta,L\to \infty$ and $\theta \to 0$ limits). More generally, Eqs. (\ref{eq:Klip}), (\ref{eq:Klim}) imply that the same limit in (\ref{eq:Klim}) is attained choosing $\theta$ such that $|\theta| / \eta \to 0$ as $\eta \to 0^{+}$, which is compatible with the parameter range for which we prove the validity of linear response, recall Eq. (\ref{eq:etatheta}). Thus, Theorem \ref{thm:response} allows to prove the validity of linear response, with longitudinal conductivity as in Theorem \ref{thm:main}.

\section{Proof of Theorem \ref{thm:main}}
\label{sec:proof}
\subsection{$\sigma_{jl}$ as right derivative of $f_{jl}$}
\label{ssec:1}
We start by rewriting the linear response coefficient $\sigma_{jl}$ as the right derivative in $0$ of the function $f_{jl}$, defined in \eqref{eqn:fjl-eta}.

\begin{proposition}
\label{prop:right-derivative}
    Under Assumption \ref{assum:H} for any $ j,l=1, 2$ we have:
    \begin{equation}
    \label{eqn:rewsch}
        \mathrm{s}_{jl}=-f_{jl}(0^+):=-\lim_{\eta \to 0^+}f_{jl}(\eta).
    \end{equation}
    As a consequence, the conductivity $\sigma_{jl}$ can be written as
    \begin{equation}
    \label{eqn:linear-response-right-derivative}
        \sigma_{jl} = \lim_{\eta \to 0^+}\frac{1}{\eta}\( f_{jl}(\eta) - f_{jl}(0^+) \).
    \end{equation}
\end{proposition}

The proof relies on estimates for the derivatives of the the Fermi projector, collected in the next lemma, whose proof is deferred to Appendix \ref{sec:aux}.

\begin{lemma}
\label{lem:norm-derivative-projector-estimate}
   Under Assumption \ref{assum:H} we have that $k\mapsto P_\mu(k)$ is $C^2$ in $\FC_1^*\setminus\{\omega_1,\ldots,\omega_n\}$. Also, there exists a finite constant $C$ such that:
 \begin{equation}
     \norm{\nabla_k P_\mu(k) }\leq C\sum_{l=1}^n\frac{1}{\abs{k-\omega_l}}\qquad\text{for all $k\in \FC_1^*\setminus\{\omega_1,\ldots,\omega_n\}$}.
 \end{equation} In particular, $\nabla_k P_\mu\in L^1(\FC_1^*)$.
\end{lemma}

\begin{proof}[Proof of Proposition \ref{prop:right-derivative}]
We shall rewrite the function $f_{jl}(\eta)$ by using an integration by parts in time $t$ and then prove identity \eqref{eqn:rewsch}.
First notice that $[P_\mu(k),H(k)]=0$ for all $k\in\FC_1^*$. 
In view of Assumption \ref{assum:H} and Lemma \ref{lem:norm-derivative-projector-estimate}, $H$ and $P_\mu$ are $C^2$ outside the Fermi points, thus by differentiating the previous identity we obtain that 
\begin{equation}
\label{eqn:dercommHP}
[ P_\mu(k),J_l(k)]=[P_\mu(k),\partial_l H(k)] = - [ \partial_l P_\mu(k),H(k)]\qquad\text{for any $k\in \FC_1^*\setminus\{\omega_1,\dots,\omega_n\}$.}
\end{equation} 
Therefore, for every $k\in \FC_1^*\setminus\{\omega_1,\dots,\omega_n\}$ we can rewrite the trace in \eqref{eqn:fjl-eta} as 
\begin{align*}
&\Tr\(\,J_j(k) \left[P_\mu(k),\tau_t(J_l(k)  )\right] \,\) 
= 
-i\frac{d}{d t}\Tr\( J_j(k) \tau_t\(\partial_l P_\mu(k)\)\).
\end{align*}

Hence, by performing an integration by parts, we get that
\begin{equation}
\label{eqn:rewfjl}
\begin{aligned}
    f_{jl}(\eta) &
   =
   \left.\frac{1}{(2\pi)^2}\int_{\FC_1^*}dk\,e^{\eta t}\Tr\( J_j(k)\, \tau_t\(\partial_l P_\mu(k)\) \,\)\right|_{-\infty}^{0}\\
   &\qquad-\frac{1}{(2\pi)^2}\int_{\FC_1^*}dk\,\int_{-\infty}^0\, dt\, \eta e^{\eta t}\Tr\( J_j(k)\, \tau_t\(\partial_l P_\mu(k)\) \).
\end{aligned}
\end{equation}
Notice that the contributions coming from the boundary terms in time are well-defined since $\partial_l P_\mu\in L^1(\FC_1^*)$ by Lemma \ref{lem:norm-derivative-projector-estimate}. In particular, the term at $t=-\infty$ vanishes, while for the term at $t=0$ we observe that
\begin{equation}
\label{eqn:firsttermfjl}
\int_{\FC_1^*}dk\Tr\( J_j(k)\, \partial_l P_\mu(k) \,\)=
-\int_{\FC_1^*}dk \Tr\( \partial^2_{jl}H(k)\,  P_\mu(k)\)
\end{equation}
where we have used an integration by parts in $k$.
Thus, to conclude that $\mathrm{s}_{jl}=-f_{jl}(0^+)$ we are left to show that 
\begin{equation}
\label{eqn:Setato0}    
\mathrm{R}(\eta):=\frac{1}{(2\pi)^2}\int_{\FC_1^*}dk\,\int_{-\infty}^0\, dt\, \eta e^{\eta t}\Tr\( J_j(k)\, \tau_t\(\partial_l P_\mu(k)\) \,\)\to 0\qquad \text{as $\eta \to 0^+$}.
\end{equation}

Observe that the operator $\partial_j P_\mu$ is off-diagonal with respect to the decomposition induced by $P_\mu$, namely:
\[
\partial_jP_\mu(k) = P_\mu(k)\partial_jP_\mu(k)(\Id-P_\mu(k)) +\text{adj} \qquad \text{for every $k \in \FC_1^*\setminus \{\omega_1,\ldots, \omega_n\}$}
\]
where \virg{$+$adj} means that the adjoint of the sum of all operators to the left is added.
Therefore, in view of the spectral decomposition, by using Remark \ref{rem:assum}\ref{it:spectraldec} and denoting by $P_j(k)$ the projector associated with $\Lambda_j(k)$ for any $1\leq j \leq N$, we have that
\begin{align*}
\tau_t\( \partial_jP_\mu(k)\)&=
 e^{iH(k)t}\partial_j P_\mu(k)e^{-iH(k)t} \\
 &=\sum_{l=1}^m\sum_{q=m+1}^{N} e^{i(\Lambda_l(k) - \Lambda_q(k))t}P_l(k)\partial_j P_\mu(k) P_q(k) + \text{adj}.
\end{align*} 
Thus, we get that
\begin{equation*}
\begin{split}
    \mathrm{R}(\eta) &= 
     \re \sum_{l=1}^m\sum_{q=m+1}^{N} \frac{1}{(2\pi)^2}\int_{\FC_1^*} dk \frac{\eta}{\eta + i(\Lambda_l(k) - \Lambda_q(k))} \Tr\big( J_l(k)P_l(k)\partial_j P_\mu(k) P_q(k)\big),
\end{split}
\end{equation*}
where $\re(\,\cdot\,)$ denotes the real part of the scalar to which it is applied. In view of Assumption \ref{assum:H} and Lemma \ref{lem:norm-derivative-projector-estimate}, we have:
$$
\abs{\frac{\eta}{\eta + i(\Lambda_l - \Lambda_q)} \Tr\big( J_l P_l(k)\partial_j P_\mu(k) P_q(k)\big)}\leq N\max_{k\in \FC_1^*}\norm{\partial_j H(k)}\abs{\nabla_k P_\mu} \in L^1(\FC_1^*)
$$
and the claim \eqref{eqn:Setato0} follows from dominated convergence theorem.
\end{proof}

\subsection{Complex time deformation}
\label{ssec:2}
In this subsection we will prove that $f_{jj}$ in \eqref{eqn:fjl-eta} can be seen as a restriction for positive $\eta$ of the function $\fjj\colon\R\to\R$, defined in \eqref{eqn:f_jj-tilde-definition}. The proof is based on the complex-deformation argument similar to the one used in \cite{GMPhall, GLMP24}, here applied to the case of non-interacting fermions.

For the subsequent analysis, it is convenient to adopt the following notations:
\begin{equation}
    P_b(k) := P_\mu(k) \qquad\qquad P_a(k) := \Id - P_\mu(k),
\end{equation}
where $a$ stands for $a$-bove the Fermi energy and $b$ for $b$-elow (or equal to) it.
Consequently, for any fibered operator $A\colon \R^2 \to \mathcal{L}(\ell^2(\X/\Gamma))\cong\mathcal{L}\(\C^N\)$ we set:
\[
A^{ab}(k): = P_a(k)A(k)P_b(k), \qquad A^{ba}(k): = P_b(k)A(k)P_a(k).
\]

\begin{proposition}
    \label{prop:complex-deformation}
    Under Assumption \ref{assum:H}, consider, for every $\eta\in \R$:
    \begin{equation}
    \label{eqn:f_jj-tilde-definition}
        \fjj(\eta) := - \frac{2}{(2\pi)^2}\int_0^{+\infty}dt\, \cos(\eta t)\int_{\FC_1^*} dk \, \Tr(\tau_{it}(J_j(k))^{ab}J_j(k)^{ba}).
    \end{equation}
    
    The function $\fjj$ is an extension to $\eta \in \mathbb{R}$ of $f_{jj}$, namely
    $$
    \fjj(\eta) = f_{jj}(\eta)\qquad\text{ for $\eta>0$}.
    $$
    \noindent In particular, the longitudinal conductivity can be expressed as:
    \begin{equation}
    \label{eqn:longitudinal-conductivity-complex-deformation}
    \sigma_{jj} = \lim_{\eta \to 0^+}\frac{1}{\eta}\left( \fjj(\eta) - \fjj(0^+) \right).
    \end{equation}
\end{proposition}

\begin{proof}
We start by rewriting $f_{jj}$ in \eqref{eqn:fjl-eta} as:
\begin{equation}
\label{eqn:rewf}
    f_{jj}(\eta) = \frac{i}{(2\pi)^2} \int_{\FC_1^*} dk \int_{-\infty}^0dt e^{\eta t} \Tr\left(\left[ \tau_t(J_j(k)), J_j(k)  \right] P_b(k) \right),
\end{equation}
where we have used that
\begin{equation*}
    [J_j(k)P_b(k),\tau_t(J_j(k))] = J_j(k)[P_b(k),\tau_t(J_j(k))]-[\tau_t(J_j(k)),J_j(k)]P_b(k).
\end{equation*}
Then, we rewrite the trace inside the integral in \eqref{eqn:rewf} as:
\begin{align*}
    \Tr([\tau_t(J_j(k)), J_j(k)]P_b(k)) 
    &= \Tr\left(\tau_t(J_j(k))^{ba}J_j(k)^{ab}\right) - \Tr\left(J_j(k)^{ba}\tau_t(J_j(k))^{ab}\right)
\end{align*}
where we have used the cyclicity of trace.
Thus, defining
\begin{equation*}
    F(\eta) := \frac{i}{(2\pi)^2} \int_{-\infty}^0dt \, e^{\eta t}\int_{\FC_1^*} dk \Tr\left( \tau_t(J_j(k))^{ba}J_j(k)^{ab}\right),
\end{equation*}
we obtain that 
\begin{equation}
\label{eqn:fjj-real-part}
\begin{split}
    f_{jj}(\eta) 
    &=   F(\eta) + \overline{F(\eta)} = 2\mathrm{Re}\left( F(\eta) \right).
\end{split}
\end{equation}
Introducing the following function
\begin{equation*}
    \C\ni\mapsto g_\eta(z) := \frac{e^{\eta z}}{(2\pi)^2}\int_{\FC_1^*} dk\, \Tr\left( \tau_z(J_j(k))^{ba} J_j^{ab}(k)\right),
\end{equation*}
we have that
\begin{equation}
\label{eqn:fsharp-as-gsharp}
    F(\eta) =  i\int_{-\infty}^0dt\, g_\eta(t).
\end{equation}
Since $g_\eta$ is analytic, the Cauchy integral theorem implies that
\begin{equation}
\label{eqn:cauchyg}
\int_{I}dz\, g_\eta(z) + \int_{II}dz\,g_\eta(z)  + \int_{III}dz\,g_\eta(z)=\oint_\mathcal{C}dz\, g_\eta (z)=0,
\end{equation}
where $\mathcal{C}$ is the closed curve drawn in Figure \ref{fig:complex-path-deformation} with $R>1$.
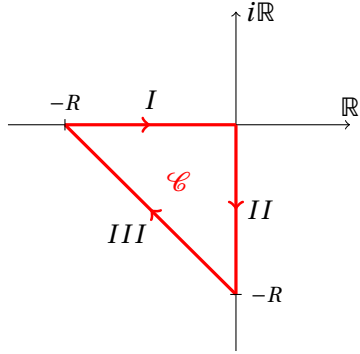
\begin{figure}
     \centering
     \usetikzlibrary{decorations.markings}
     \begin{tikzpicture}[scale=1.5]
        \draw[->] (-2,0) -- (1,0) node[anchor=south] {$\R$};
        \draw[->] (0,-2) -- (0,1) node[anchor=west] {$i\R$};
        \begin{scope}[very thick,decoration={
                markings,
                mark=at position 0.5 with {\arrow{>}}}
                ]
        \draw[postaction={decorate},color=red] (-1.5,0)-- node[above=2pt] {{\color{black} $I$}} (0,0);
        \draw[postaction={decorate},color=red] (0,0)-- node[anchor=west] {{\color{black} $II$}} (0,-1.5); 
        \draw[postaction={decorate},color=red] (0,-1.5)-- node[xshift=-9pt,yshift=-9pt] {{\color{black} $III$}} (-1.5,0);
        \node[color=red] at (-.5,-.5) {$\mathcal{C}$};
        \end{scope}
        \draw (-1.5,-.05) -- (-1.5,.05) node[anchor=south] {\footnotesize $-R$};
        \draw (-.05,-1.5) -- (.05,-1.5) node[anchor=west] {\footnotesize  $-R$};
     \end{tikzpicture}
    \caption{Path used in the complex-time deformation.}
    \label{fig:complex-path-deformation}
\end{figure}
It is easy to see that
\begin{equation}
\label{eqn:cauchy-complex-deformation}
    \lim_{R\to \infty}\int_{III}dz\,g_\eta (z)  =0.
\end{equation}
Therefore, taking the limit $R\to\infty$ of \eqref{eqn:cauchyg} and using \eqref{eqn:fsharp-as-gsharp} we get that
\begin{equation}
\label{eqn:fA-complex-time}
\begin{aligned}
     F(\eta)
    &=- \frac{1}{(2\pi^2)}\int^{+\infty}_0dt\int_{\FC_1^*} dk e^{-i\eta t}\Tr\left(\tau_{it}(J_j(k))^{ab} J_j^{ba}(k)\right),
\end{aligned}
\end{equation}
where 
we have parametrized $II$ as $\gamma:[0,R]\ni t \mapsto -it$ and 
we have exploited the cyclicity of trace.
Computing the real part of $F(\eta)$, the proof is concluded.
\end{proof}

\begin{remark}
\label{rmk:complex-deformation-k-domain}
    Let $B\subset \Bril$ be an open set. Notice that in the proof of Proposition \ref{prop:complex-deformation} the complex time deformation holds also if in the expression of $\fjj$ we replace the integration domain $\Bril$ with $B$.
\end{remark}

Obviously $\fjj$ is an even function. If $\fjj$ was differentiable, then its derivative in zero would vanish; hence, by Proposition \ref{prop:complex-deformation} $\sigma_{jj}$ would be zero. As we will see, differentiability in zero does not hold. However, we will use this observation to get rid of all the contributions to the conductivity associated with energies away from the Fermi level. We will be left with the contributions due to the Fermi points, which we will be able to explicitly determine.

\subsection{Singular and regular parts of $\fjj$}
\label{sec:reular-part-fjj}
\begin{figure}
    \centering
    \begin{tikzpicture}[scale=1.5]
        \draw[->] (-1,0) -- (3.5,0);
        \draw[->] (0,-.5) -- (0,1.5);
        \draw[dashed] (1,-.1) node[anchor=north] {$1$} -- (1,1);
        \draw[dashed] (2,-.1) node[anchor=north] {$2$} -- (2,1);
        \draw[very thick, color=orange] (0,1)  -- (1,1) node[above] {$\chi$}  .. controls (1.5,1) and (1.5,0) .. (2,0) -- (3.4,0);
        \draw[very thick, color=blue] (0,0) -- (1,0)  .. controls (1.5,0) and (1.5,1) .. (2,1) node[above] {$1-\chi$} -- (3.4,1);
    \end{tikzpicture}
    \caption{Graphs of the smooth cut-off functions $\chi$, $1-\chi$.}
    \label{fig:chi-example}
\end{figure}
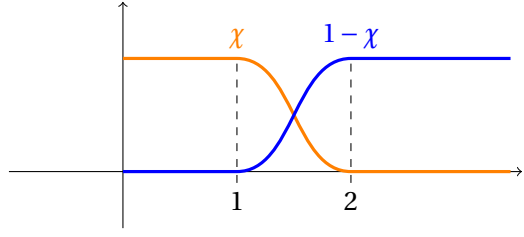

First we single out from $\fjj$ the singular part, denoted by $\fjjsing$, which is due to the energies close to the Fermi energy $\mu$. We proceed as follows.

\begin{enumerate}[label=(\roman*), ref=(\roman*)]
\item\label{it:delta} Let $\delta>0$ small enough, and let $\chi \in C^\infty(\R_+)$ such that $\chi(x) = 1$ if $x \leq 1$ and $\chi(x) = 0$ if $x \geq  2$, see \eg Figure \ref{fig:chi-example}. We define:
\begin{equation}
\label{eqn:chi<-definition}
    \chi_<(k) := \chi\left(\delta^{-1}|H(k) - \mu|\right) \qquad \text{ and }\quad \chi_>(k): = \Id - \chi_<(k);
\end{equation}
The function $\chi_{>}(k)$ introduces a smooth cutoff, supported away from the Fermi energy. 
\item \label{it:eps} Let $\eps >0$ small enough such that expansions \eqref{eqn:conical} for $\Lambda_{\pm}$ hold true and the following condition is satisfied.
Let:
\begin{equation}
    d_F:=\min\left\{|S_k(\omega_l-\omega_j)| \,:\, j,k,l=1,\ldots,n\right\}.
\end{equation}

Let $\eps < d_F$. We define the set $B_\eps$ as
\begin{equation}
\label{eqn:definition-Bepsilon}
    B_\eps := \bigcup_{l=1}^n B_\eps^{(l)} \quad \text{ with } \quad B_\eps^{(l)} := \left\{ k \in \Bril \,:\, 2|S_l(k-\omega_l)|< \eps \right\}.
\end{equation}

That is, $B_\eps$ is the union of $n$ disjoint open sets, each of them containing exactly one Fermi point.
\end{enumerate}

We then define the singular and the regular parts of $\fjj$ as:
    \begin{align}
    \label{eqn:fjjsing-def}
    \begin{split}
        \fjjsing(\eta) &:= - \frac{2}{(2\pi)^2}\int_0^{+\infty}dt\,\cos(\eta t)\int_{B_\eps} dk\,\Tr\left(\tau_{it}\left((\chi_<J_j\chi_<)(k)\right)^{ab}\left(\chi_<J_j\chi_<\right)^{ba}(k)\right)
    \end{split} \\
        \fjjreg(\eta) &:= \fjj(\eta) - \fjjsing(\eta).\label{eqn:fjjreg}
    \end{align}

\begin{proposition}
\label{prp:regular-part-fjj}
    Under Assumption \ref{assum:H}, the map $\fjjreg$ is even and differentiable. In particular:
    \begin{equation}
        \label{eqn:longitudinal-conductivity-singular}
        \sigma_{jj} = \lim_{\eta \to 0+}\frac{1}{\eta}\left( \fjjsing(\eta) - \fjjsing(0^+)\right).
    \end{equation}
\end{proposition}
\begin{proof}
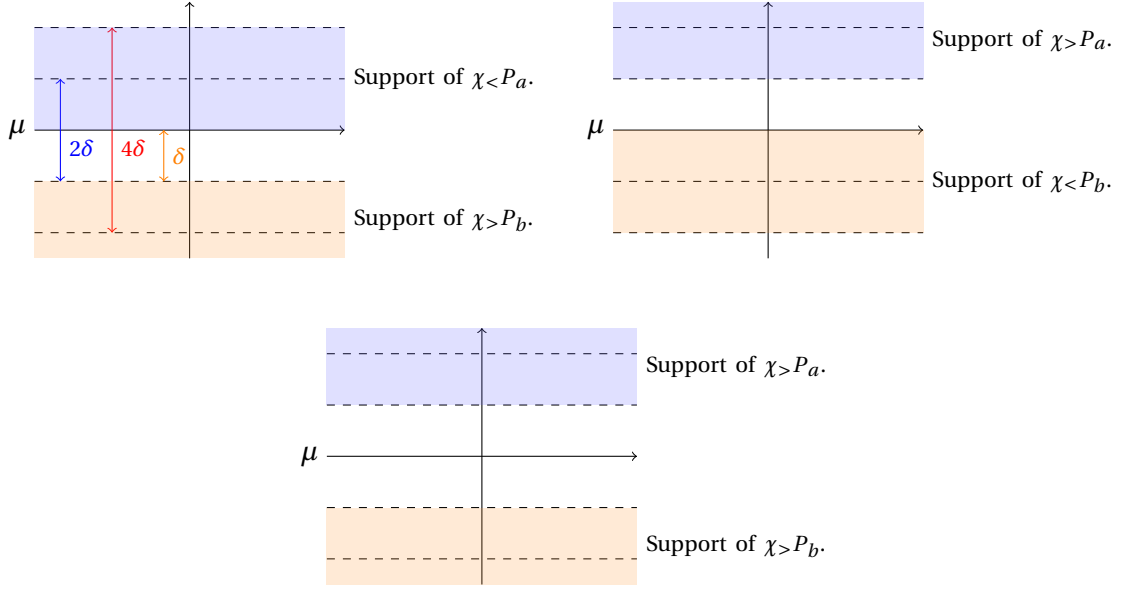
\begin{figure}
    \centering
    \begin{subfigure}
        \centering
        \begin{tikzpicture}[scale=0.68]
            \draw[->] (-3,0) node[anchor=east] {$\mu$} -- (3,0);
            \draw[->] (0,-2.5) -- (0,2.5);
            \foreach \n in {-2,-1,1,2}
                \draw[dashed] (-3,\n) -- (3,\n);
            \draw[<->, color=blue] (-2.5,-1) -- node[anchor=north west] {\footnotesize $2\delta$} (-2.5,1);
            \draw[<->, color=red] (-1.5,-2) -- node[anchor=north west] {\footnotesize $4\delta$} (-1.5,2);
            \draw[<->, color=orange] (-0.5,-1) -- node[anchor=west] {\footnotesize $\delta$} (-0.5,0);
            \path[fill=blue!60, opacity=0.2] (-3,0) -- ++(6,0) -- node[anchor=west, opacity=1, text width=24mm] {\color{black} \footnotesize Support of $\chi_<P_a$.} ++(0,2) -- ++(-6,0) -- cycle;
            \path[fill=orange!80, opacity=0.2] (-3,-1) -- ++(6,0) -- node[anchor=west, opacity=1, text width=24mm] {\color{black} \footnotesize Support of $\chi_>P_b$.} ++(0,-1.5)  -- ++(-6,0) -- cycle;
        \end{tikzpicture}
    \end{subfigure}
    \hspace{.2cm}
    \begin{subfigure}
        \centering
        \begin{tikzpicture}[scale=0.68]
            \draw[->] (-3,0) node[anchor=east] {$\mu$} -- (3,0);
            \draw[->] (0,-2.5) -- (0,2.5);
            \foreach \n in {-2,-1,1,2}
                \draw[dashed] (-3,\n) -- (3,\n);
            \path[fill=blue!60, opacity=0.2] (-3,1) -- ++(6,0) -- node[anchor=west, opacity=1, text width=24mm] {\color{black} \footnotesize Support of $\chi_>P_a$.} ++(0,1.5) -- ++(-6,0) -- cycle;
            \path[fill=orange!80, opacity=0.2] (-3,0) -- ++(6,0) -- node[anchor=west, opacity=1, text width=24mm] {\color{black} \footnotesize Support of $\chi_<P_b$.} ++(0,-2)  -- ++(-6,0) -- cycle;
        \end{tikzpicture}
    \end{subfigure}
    
    \vspace{.7cm}
    \begin{subfigure}{}
        \centering
        \begin{tikzpicture}[scale=0.68]
            \draw[->] (-3,0) node[anchor=east] {$\mu$} -- (3,0);
            \draw[->] (0,-2.5) -- (0,2.5);
            \foreach \n in {-2,-1,1,2}
                \draw[dashed] (-3,\n) -- (3,\n);
            \path[fill=blue!60, opacity=0.2] (-3,1) -- ++(6,0) -- node[anchor=west, opacity=1, text width=24mm] {\color{black} \footnotesize Support of $\chi_>P_a$.} ++(0,1.5) -- ++(-6,0) -- cycle;
            \path[fill=orange!80, opacity=0.2] (-3,-1) -- ++(6,0) -- node[anchor=west, opacity=1, text width=24mm] {\color{black} \footnotesize Support of $\chi_>P_b$.} ++(0,-1.5)  -- ++(-6,0) -- cycle;
        \end{tikzpicture}
        \begin{subfigure}{}
        \centering
        \begin{tikzpicture}[scale=0.68]
            \draw[->] (-3,0) node[anchor=east] {$\mu$} -- (3,0);
            \draw[->] (0,-2.5) -- (0,2.5);
            \foreach \n in {-2,-1,1,2}
                \draw[dashed] (-3,\n) -- (3,\n);
            \path[fill=blue!60, opacity=0.2] (-3,0) -- ++(6,0) -- node[anchor=west, opacity=1, text width=24mm] {\color{black} \footnotesize Support of $\chi_<P_a$.} ++(0,2) -- ++(-6,0) -- cycle;
            \path[fill=orange!80, opacity=0.2] (-3,0) -- ++(6,0) -- node[anchor=west, opacity=1, text width=24mm] {\color{black} \footnotesize Support of $\chi_<P_b$.} ++(0,-2)  -- ++(-6,0) -- cycle;
        \end{tikzpicture}
    \end{subfigure}
    \end{subfigure}
    \caption{Graphical representation of the energy cutoff. The colored regions represent the region of the spectrum selected by $\chi_{\alpha_1} P_\pm$. In all the terms that we estimate a gap appears, which allows us to write the inequality \eqref{eqn:ineqwithone>}. }
    \label{fig:energy-cut}
\end{figure}

For every $\alpha_1,\alpha_2\in \{>,<\}$, let us define
\begin{equation}
\label{eqn:Ialphabeta}
I_{\alpha_1,\alpha_2}(\eta):=\int_0^{+\infty}dt\, \cos(\eta t)\int_{\FC_1^*} dk \, \Tr\left(\tau_{it}\left((\chi_\alpha J_j\chi_\beta)(k)\right)^{ab}\left(J_j\right)(k)^{ba}\right).
\end{equation}
Thus, we have that
\begin{equation}
\label{eqn:splittingf}
\fjj(\eta)=I_{<,<}(\eta)+I_{<,>}(\eta)+I_{>,<}(\eta)+I_{>,>}(\eta).
\end{equation}
We shall prove that the last three summands in \eqref{eqn:splittingf} are differentiable; since they are even functions, they do not contribute to $\sigma_{jj}$. In view of the spectral decomposition for $e^{H(k)t}$ and the definition of $\chi_<$ and $\chi_>$, it is easy to see that there is a gap of at least $\delta$ in the eigenvalues that remain in the expression of the last three terms. This is illustrated in Figure \ref{fig:energy-cut}.
More specifically we obtain the following inequality for some constant $C>0$
\begin{equation}
\label{eqn:ineqwithone>}    
\norm{t^p\tau_{it}\left((\chi_\alpha J_j\chi_\beta)(k)\right)^{ab}}\leq C e^{-\frac{\delta}{2} t}
\end{equation}
for all $\alpha_1, \alpha_2$: $\alpha_i$ equals $>$, for any $t>0$ and $p\in\{0,1\}$. Lebesgue's dominated convergence theorem implies that $I_{<,>}, I_{>,<}, I_{>,>}$ are $C^1$ function in $\eta\in\R$. Hence, only $I_{<,<}$ contributes to the conductivity $\sigma_{ll}$; we want to show that moments inside $B_\eps$ gives contribution to $\sigma_{jj}$. 
This can be done by a similar argument, using that outside of $B_\eps$ the eigenvalues that remain in the expression have a spectral gap, this time of length $c\eps$ for some constant $c>0$.
Therefore only $I_{<,<}$ where the integral in $k$ is performed only on $B_\eps$ contribute to $\sigma_{jj}$, proving the claim.
\end{proof}

\subsection{Final rewriting of $\sigma_{jj}$}
\label{sec:explicit-conductivity}
In this last subsection, we impose two additional conditions on the small positive parameters $\delta$ and $\eps$, introduced previously respectively in Subsection \ref{sec:reular-part-fjj}\ref{it:delta} and \ref{it:eps}.

First, letting $\Lambda_1(k),\ldots,\Lambda_N(k)$ the eigenvalues of $H(k)$, we choose $\eps\equiv\eps_\delta$ small enough so that 
\begin{equation}
\label{eqn:condition-eps-delta}
\text{if} \quad   k\in B_\eps\quad \text{ then }\quad \begin{cases}
    |\Lambda_j(k) -\mu| < \delta & \quad \text{if } \Lambda_j(k) = \Lambda_-(k) \text{ or } \Lambda_j(k) = \Lambda_+(k) \\
    |\Lambda_j(k) -\mu| > 2\delta & \quad \text{otherwise}.
    \end{cases}
\end{equation}
This condition implies that if $k\in B_\eps$ then  $\chi_<(k) = P_-(k) + P_+(k)$, where $P_-$ and $P_+$ are the projectors associated to $\Lambda_-$ and $\Lambda_+$ respectively.

Second, by Assumption \ref{assum:H}, it is possible to choose $\eps$ sufficiently small in order to have two constants $m,M>0$ such that 
\begin{equation}
\label{eqn:condition-cone-eps}
    m\abs{k-\omega_l}\leq \Lambda_+(k) - \Lambda_-(k) \leq M\abs{k-\omega_l}\quad\text{for any $k\in B_\eps^{(l)}$.}
\end{equation}

Moreover, at this point it is useful to point out that in view of Assumption \ref{assum:H}\ref{item:conical} and by using Remark \ref{rem:assum}\ref{it:eigenvalcont} the map 
\begin{equation}
\label{eqn:sumPpm}
B_\eps\ni k\mapsto P_-(k)+P_+(k)\,\text{ is $C^2$},
\end{equation}
since the spectral subset $\{\Lambda_-(k)\cup\Lambda_+(k)\}$ is uniformly separated (by a distance at least $\delta$ assuming condition \eqref{eqn:condition-eps-delta}) by the rest of the spectrum of $H(k)$ for all $k\in B_\eps$ . 

Starting from Proposition \ref{prp:regular-part-fjj} we shall prove the final expression for the longitudinal conductivity $\sigma_{jj}$ given by \eqref{eqn:main}, by showing formulas being increasingly more explicit, and involving step by step only the expansions of the two eigenvalues around the Fermi points \eqref{eqn:conical}.
To begin, we perform the time integration; this is the content of the next proposition.

\begin{proposition}
\label{prp:hjj-two-eigenvalues}
    Under Assumption \ref{assum:H}, let $\fjjsing$ be defined as in \eqref{eqn:fjjsing-def} and let $P_-(k)$, $P_+(k)$ be the projectors associated with the eigenvalues $\Lambda_-(k)$, $\Lambda_+(k)$ respectively. Choosing $\eps$ and $\delta$ so that condition \eqref{eqn:condition-eps-delta} holds, we have:
    \begin{equation}
    \label{eqn:fjjsing-two-eigenvalues-current}
    \fjjsing(\eta)=\frac{2}{(2\pi)^2}\int_{B_\eps} dk \frac{(\Lambda_-(k) - \Lambda_+(k))}{\eta^2 + (\Lambda_-(k) - \Lambda_{+}(k))^2} \Tr\big(  P_-(k)J_j(k)P_+(k)J_j(k)P_-(k)\big).
    \end{equation}
\end{proposition}

At the second step, we rewrite the expression for $\sigma_{jj}$ involving only the eigenvalues $\Lambda_{\pm}$ and their derivatives up to the second order around the Fermi points.

\begin{proposition}
\label{prp:longitudinal-conductivity-eigenvalues}
    Under Assumption \ref{assum:H}, we define the following function depending on $\eta>0$:
    \begin{equation}
    \label{eqn:zjj-definition}
    \begin{split}
        \zjj(\eta) &:= \frac{1}{(2\pi)^2}\int_{B_\eps} dk \frac{\Lambda_-(k) - \Lambda_+(k)}{\eta^2 + (\Lambda_+(k) - \Lambda_-(k))^2}\cdot \\
        & \quad\quad\cdot\left( \frac{1}{2}\partial_j^2 \left( (\Lambda_+(k) -\mu)^2 + (\Lambda_-(k) - \mu)^2  \right)   -(\partial_j\Lambda_+(k))^2 - (\partial_j\Lambda_-(k))^2\right).
    \end{split}
    \end{equation}
    Then, by choosing $\eps$ satisfying  conditions \eqref{eqn:condition-eps-delta} and \eqref{eqn:condition-cone-eps}: 
    \begin{equation}
    \label{eqn:lingitudinal-conductivity-eigenvalues}
        \sigma_{jj} = \lim_{\eta\to 0^+}\frac{1}{\eta}(\zjj(\eta) - \zjj(0^+)).
    \end{equation}
\end{proposition}

The proofs of Propositions \ref{prp:hjj-two-eigenvalues}, \ref{prp:longitudinal-conductivity-eigenvalues} are deferred to Appendix \ref{sec:aux}.

We would like to express $\sigma_{jj}$ by only using the linearization of the eigenvalues around the Fermi energy. This is a delicate point, since the expression  \eqref{eqn:lingitudinal-conductivity-eigenvalues} includes the second order derivatives of the eigenvalues $\Lambda_\pm$. This issue is circumvented integrating by parts, as the next proposition shows.

\begin{proposition}
\label{prp:linearization-conductivity}
    Let Assumption \ref{assum:H} hold true, and let $\eps$ satisfy  conditions \eqref{eqn:condition-eps-delta} and \eqref{eqn:condition-cone-eps}. Then the longitudinal conductivity $\sigma_{jj}$ can be written as
    \begin{equation}
    \label{eqn:conductivity -linearization}
        \sigma_{jj} =\frac{1}{(2\pi)^2} \sum_{l=1}^n\lim_{\eta\to 0^+}\eta \int_{B_\eps^{(l)}}dk \frac{\partial_j^2\big( |S_l(k-\omega_l)| \big)}{\eta^2 +4 |S_l(k-\omega_l)|^2}\,,
    \end{equation}
where $\omega_l$ is the $l$-th Fermi point, $S_l$ is its associated matrix as in Assumption \ref{assum:H}\ref{item:conical} and $B_\eps^{(l)}$ is a neighborhood of $\omega_l$ as defined in \eqref{eqn:definition-Bepsilon}.    
\end{proposition}

\begin{proof}
It is useful to introduce the following auxiliary functions. For any $k\in B_\eps$:
\begin{equation*}
    s(k):= (\Lambda_+(k)-\mu) + (\Lambda_-(k)-\mu) \qquad d(k) := (\Lambda_+(k)-\mu) - (\Lambda_-(k)-\mu)\equiv\Lambda_+(k)-\Lambda_-(k),
\end{equation*}
in terms of which:
\begin{equation}
\label{eqn:inversion-s-d}
    \Lambda_+(k) - \mu = \frac{s(k) + d(k)}{2}, \qquad \Lambda_-(k) - \mu = \frac{s(k) - d(k)}{2}.
\end{equation}
Notice that $B_\eps\ni k\mapsto s(k)=\Tr(( H(k)-\mu)(P_-(k)+P_+(k)))$ is $C^2$ by using \eqref{eqn:sumPpm} and Assumption \ref{assum:H}\ref{item:perselfadj}. On the other hand, the function $B_\eps\ni k\mapsto d(k)$ is not $C^1$ since $\Lambda_\pm$ are in general not $C^1$. Furthermore, we claim that:
\begin{equation}
\label{eqn:inequalities-derivative-d}
\begin{split}
    |\partial_j d(k)| &\leq C  \\
    |\partial_j^2 d(k)| &\leq  \frac{C}{|k-\omega_l|}
\end{split}
\qquad \text{for all } k \in B_\eps^{(l)}\setminus\{\omega_l\},
\end{equation}
where the neighborhood $B_\eps^{(l)}$ of the Fermi point $\omega_l$ is defined in \eqref{eqn:definition-Bepsilon}.
We shall obtain the above inequalities by proving them for the eigenvalues $\Lambda_\pm$. Let us analyze only the eigenvalue $\Lambda_-$, the same argument applies to $\Lambda_+$. In view of the identity $\Lambda_-(k)=\Tr(H(k)P_-(k))$ we have that, for every $k \in B_\eps \setminus\{\omega_1 ,\ldots, \omega_n \}$:
\begin{equation}
\label{eqn:derLambda-}
\partial_j\Lambda_-(k)=\Tr(\partial_j H(k) P_-(k)).
\end{equation}
Thus, the first inequality in \eqref{eqn:inequalities-derivative-d} follows from \eqref{eqn:derLambda-}. 
Let us now prove the second inequality in \eqref{eqn:inequalities-derivative-d}. To this end, we observe that, for every $k \in B_\eps^{(l)} \setminus\{\omega_l \}$: 
\[
\abs{\partial^2_j\Lambda_-(k)}\leq \abs{\Tr(\partial^2_j H(k)P_-(k))}+\abs{\Tr(\partial_j H(k)\partial_j P_-(k))}\leq \frac{C}{|k-\omega_l|}, 
\]
where we have applied Lemma \ref{lem:norm-derivative-projector-estimate}. Thus, the second inequality in \eqref{eqn:inequalities-derivative-d} holds. 

By Proposition \ref{prp:longitudinal-conductivity-eigenvalues} the longitudinal conductivity is given by $\sigma_{jj} = \lim_{\eta\to 0^+}\frac{1}{\eta}(\zjj(\eta) - \zjj(0^+))$, where the function $\zjj(\eta)$ is defined as in \eqref{eqn:zjj-definition}; we shall now rewrite this expression in terms of $s(k)$ and $d(k)$. For any $f,g$, whose first and second partial derivatives exist, we have:
\begin{equation}
\label{eqn:dersumsquares}
 \frac{1}{2}\partial_j^2(f^2 +g^2) - (\partial_j f)^2 -  (\partial_j g)^2 =  f\partial_j^2 f +  g\partial_j^2 g.
\end{equation}
Thus:
\begin{equation*}
\frac{1}{2}\partial_j^2 \left( (\Lambda_+(k) -\mu)^2 + (\Lambda_-(k) - \mu)^2  \right)   -(\partial_j\Lambda_+(k))^2 - (\partial_j\Lambda_-(k))^2
=\frac{1}{2}\(s\,\partial_j^2s +d\,\partial_j^2 d \),
\end{equation*}
where in the second equality we have used \eqref{eqn:inversion-s-d}. Therefore, we obtain:
\begin{equation*}
 \zjj(\eta) = -\frac{1}{2(2\pi)^2}\int_{B_\eps} dk \frac{d(k)}{\eta^2 + d^2(k)}\(s(k)\,\partial_j^2s(k) +d(k)\,\partial_j^2 d(k) \)   
\end{equation*}
which implies:
\begin{equation}
\label{eqn:sigmainsed}
\sigma_{jj}=\frac{1}{2}\frac{1}{(2\pi)^2}\lim_{\eta \to 0^+} \(\eta \int_{B_\eps}dk\frac{s(k)\partial_j^2 s(k)}{\left(\eta^2 + d^2(k)\right)d(k)} + \eta\int_{B_\eps}dk\frac{\partial_j^2 d(k)}{\eta^2 + d^2(k)} \).   
\end{equation}
It is easy to see that the contribution coming from the first term in the right-hand side of \eqref{eqn:sigmainsed} vanishes, by using $\abs{s(k)}\leq C\abs{k-\omega_l}$ in view of \eqref{eqn:conical}, $\abs{\partial_j^2 s(k)}\leq C$ in view of its $C^2$ regularity and \eqref{eqn:condition-cone-eps} for every $k\in B_\eps^{(l)}$.
To conclude, we shall prove that only the linearization of $d(k)$ contributes to $\sigma_{jj}$. To this end, in view of \eqref{eqn:conical} we observe that  
\begin{equation}
\label{eqn:d(k)}
d(k)=2\abs{S_l(k-\omega_l)}+r_l(k-\omega_l)\qquad\text{for every $k\in B_\eps^{(l)}$},    
\end{equation}
where $r_l(k-\omega_l)=o(\abs{k-\omega_l})$. Therefore:
\begin{equation*}
\frac{1}{\eta^2 + d^2(k)}=\frac{1}{\eta^2 + 4 \abs{S_l(k-\omega_l)}^2}+\frac{R_l(k-\omega_l)}{(\eta^2 + d^2(k))(\eta^2 + 4 \abs{S_l(k-\omega_l)}^2)},   
\end{equation*}
where $R_l(k-\omega_l)=o(\abs{k-\omega_l}^2)$ for any $k\in B_\eps^{(l)}$. 
We have that:
\begin{equation}
\label{eqn:linsecterm}
\lim_{\eta\to 0^+}\eta\int_{B_\eps^{(l)}}dk\frac{\partial_j^2 d(k) R_l(k-\omega_l)}{(\eta^2 + d^2(k))(\eta^2 + 4 \abs{S_l(k-\omega_l)}^2)}=0,
\end{equation}
using \eqref{eqn:condition-cone-eps}, \eqref{eqn:inequalities-derivative-d} and dominated convergence theorem.
Thus, we obtain that:
\begin{equation}
\label{eqn:sigmainsed2}
\sigma_{jj}=\frac{1}{2}\frac{1}{(2\pi)^2}\sum_{l=1}^n\lim_{\eta \to 0^+}  \eta\int_{B_\eps^{(l)}}dk\frac{\partial_j^2 d(k)}{\eta^2 + 4 \abs{S_l(k-\omega_l)}^2}.   
\end{equation}
To conclude, we are left with showing that we can replace $d(k)$ at the numerator with its linearization. To this end, we perform a double integration by parts, proving that the boundary terms do not contribute to $\sigma_{jj}$. We have:
\begin{equation}
\label{eqn:first-integration-parts}
\begin{split}
	 \int_{B_\eps^{(l)}}dk \frac{\partial_j^2 d(k)}{\eta^2 + 4|S_l(k-\omega_l)|^2} &= 
    \int_{\partial B_\eps^{(l)}} d\sigma (k) \frac{\nu_j \partial_j d(k)}{\eta^2 + 4|S_l(k-\omega_l)|^2} \\
    &\quad - \int_{\partial B_\eps^{(l)}}d\sigma (k)\, \nu_j \partial_j\left(\frac{1}{\eta^2 + 4|S_l(k-\omega_l)|^2}\right) d(k)\\
    &\quad+ \int_{B_\eps^{(l)}}dk \,\partial_j^2\left(\frac{1}{\eta^2 + 4|S_l(k-\omega_l)|^2}\right)d(k),
\end{split}
\end{equation}
where $\nu_j$ is the $j$-th component of the outward unit normal vector to $\partial B_\eps^{(l)}$ and $d\sigma (k)$ is the line element. 
It is easy to show that both the boundary terms do not contribute to $\sigma_{jj}$ in the limit $\eta \to 0^+$, by using that $\partial_j \abs{S_l(k-\omega_l)}\leq C$ and the first inequality in \eqref{eqn:inequalities-derivative-d}.
Now substituting \eqref{eqn:d(k)} in the last term of \eqref{eqn:first-integration-parts} only the part containing $|2S_l(k-\omega_l)|$ in the numerator gives contribution to $\sigma_{jj}$ in the adiabatic limit.
Therefore we obtain:
\begin{equation}\label{eq:fin}
\sigma_{jj}=\frac{1}{(2\pi)^2}\sum_{l=1}^n\lim_{\eta \to 0^+}  \eta \int_{B_\eps^{(l)}}dk \,\partial_j^2\left(\frac{1}{\eta^2 + 4|S_l(k-\omega_l)|^2}\right)\abs{S_l(k-\omega_l)}.
\end{equation}
Performing again double integration by parts concludes the proof.
\end{proof}

The next lemma concludes the proof of Theorem \ref{thm:main}.

\begin{lemma}
\label{lem:final}  
Under the hypotheses of Proposition \ref{prp:linearization-conductivity}, we have:
\begin{equation}\label{eq:ff}
\sigma_{jj} = \frac{1}{16}\sum_{l=1}^n \frac{s_{l,1j}^2 + s_{l,2j}^2}{|\det S_l|}.
\end{equation}
\end{lemma}

\begin{proof}[Proof of Lemma \ref{lem:final}]
For the sake of readability, we shall omit the $l$-dependence of the matrix element $s_{ij}\equiv s_{l,ij}$. In view of equality \eqref{eqn:conductivity -linearization}, we compute:
\begin{equation*}
\partial_j^2\left( |S_l k| \right) = \sum_{\alpha =1}^2 \left(\frac{(s_{\alpha j})^2}{|S_l k|} - \frac{(S_l k)_\alpha^2 (s_{\alpha j})^2}{|S_l k|^3} \right)  - \sum_{\alpha\neq \beta}^{1,2}\frac{(S_l k)_\alpha(S_l k)_\beta (s_{\alpha j})(s_{\beta j})}{|S_l k|^3},
\end{equation*}
where $(S_l k)_\alpha$ is the $\alpha$-th component of the vector $S_l k$.
Thus,
\begin{equation}
\label{eqn:main1}
\begin{aligned}
\int_{B_\eps^{(l)}}dk \frac{\partial_j^2\big( |S_l(k-\omega_l)| \big)}{\eta^2 +4|S_l(k-\omega_l)|^2} 
    =& \frac{1}{2\abs{\det S_l}}\int_{\abs{k}<\eps} dk \frac{1}{\eta^2 + \abs{k}^2} \sum_{\alpha =1}^2 \left(\frac{(s_{\alpha j})^2}{\abs{k}} - \frac{k_\alpha^2 (s_{\alpha j})^2}{\abs{k}^3} \right) \\
	& -\frac{1}{2\abs{\det S_l}} \int_{\abs{k}<\eps} dk \frac{1}{\eta^2 + \abs{k}^2}\sum_{\alpha\neq \beta}^{1,2}\frac{k_\alpha k_\beta s_{\alpha j}s_{\beta j}}{\abs{k}^3},  
\end{aligned}
\end{equation}
where in the last equality we have used the change of variable $k\mapsto 2S_l(k-\omega_l)$, recalling also that $B_\eps^{(l)} = \left\{ k\in \Bril \, : \, 2|S_l(k-\omega_l)| < \eps    \right\}$.
Observe that the second term on the right-hand side of \eqref{eqn:main1} vanishes since the integrand is odd in $k$. Concerning the first term, we can rewrite it as:
\begin{equation*}
\frac{1}{2\abs{\det S_l}}\int_{\abs{k}<\eps} dk \frac{1}{\eta^2 + \abs{k}^2} \sum_{\alpha =1}^2 \left(\frac{(s_{\alpha j})^2}{\abs{k}} - \frac{k_\alpha^2 (s_{\alpha j})^2}{\abs{k}^3} \right)=\frac{(s_{1j})^2+(s_{2j})^2}{4\abs{\det S_l}}\int_{\abs{k}<\eps} dk \frac{1}{(\eta^2 + \abs{k}^2)\abs{k} },  
\end{equation*}
where we have implemented the change of variable $k_1\leftrightarrow k_2$ in the second term in brackets above. In conclusion:
\begin{align*}
\sigma_{jj} 
&=\frac{1}{(2\pi)^2} \sum_{l=1}^n \frac{(s_{1j})^2+(s_{2j})^2}{4\abs{\det S_l}}\lim_{\eta\to 0^+}\eta\int_{\abs{k}<\eps} dk \frac{1}{(\eta^2 + \abs{k}^2)\abs{k} }
=\frac{1}{16}\sum_{l=1}^n \frac{(s_{1j})^2+(s_{2j})^2}{\abs{\det S_l}}.
\end{align*}
This concludes the proof of (\ref{eq:ff}).
\end{proof}

\section{Proof of Theorem \ref{thm:response}}\label{sec:proofresp}

Following \cite{GLMP24}, we start by approximating the original dynamics with a suitable auxiliary evolution. Let us denote by $\widetilde{\mathcal{H}}(t)$ the time-dependent Hamiltonian after having replaced the adiabatic parameter $\eta$ by $\eta_{\beta}$; the reason for this replacement will be discussed below. Let us denote by $\widetilde{\Gamma}(t)$ the solution of (\ref{eq:dyn}), with $\mathcal{H}(t)$ replaced by $\widetilde{\mathcal{H}}(t)$. Proceeding as in \cite{GLMP24}, Proposition 4.1, we have:
\begin{equation}\label{eq:comp}
\Big| \Tr_{\mathcal{F}} \mathcal{O}_{z} \Gamma(t) -  \Tr_{\mathcal{F}} \mathcal{O}_{z} \widetilde{\Gamma}(t) \Big|\leq  \frac{C_{\mathcal{O}}}{\eta^{4} \beta}\;.
\end{equation}
The proof of (\ref{eq:comp}) follows from a standard application of Lieb-Robinson bounds for non-autonomous dynamics. The second error term in (\ref{eq:Rest}) is produced by this comparison. Next, let us consider the Duhamel expansion for $\widetilde{\Gamma}(t)$. We have:
%
%
%
\begin{equation}\label{eq:duha}
\begin{split}
&\Tr_{\mathcal{F}} \mathcal{O}_{z} \widetilde{\Gamma}(0) - \Tr_{\mathcal{F}} \mathcal{O}_{z} \Gamma_{\beta,L} \\
&= \sum_{m\geq 1} (-i)^{m} \int_{-\infty \leq t_{m} \leq \ldots \leq t_{1} \leq 0} d\underline{t} \, \Big\langle \Big[\cdots \Big[ \Big[ \mathcal{O}_{z}, \tau_{t_{1}}(\mathcal{P}(\eta_{\beta} t_{1})) \Big], \tau_{t_{2}}(\mathcal{P}(\eta_{\beta} t_{2})) \Big] \cdots \tau_{t_{n}}(\mathcal{P}(\eta_{\beta} t_{m})) \Big] \Big\rangle_{\beta,L}\;,
\end{split}
\end{equation}
with $\langle \cdot \rangle_{\beta,L}$ the average with respect to the equilibrium state $\Gamma_{\beta,L}$. Recall that
\begin{equation}\label{eq:expP}
\mathcal{P}(\eta_{\beta} t) = \sum_{n\geq 1} \frac{e^{n \eta_{\beta} t}}{n!} \mathcal{P}_{n}^{A}\;,
\end{equation}
with $\mathcal{P}_{n}^{A}$ finite-ranged. At finite $\beta, L$, the convergence of the Duhamel series, non-uniformly in $\eta$, follows from standard arguments (see {\it e.g.} \cite{GLMP24}, Section 4.1). Our goal will be to improve the estimates, and to gain uniformity in $\eta$ and in $L$. 

Using (\ref{eq:expP}), we rewrite the argument of the integral in (\ref{eq:duha}) as:
\begin{equation}\label{eq:nth}
\begin{split}
&\Big\langle \Big[\cdots \Big[ \Big[ \mathcal{O}_{z}, \tau_{t_{1}}(\mathcal{P}(\eta_{\beta} t_{1})) \Big], \tau_{t_{2}}(\mathcal{P}(\eta_{\beta} t_{2})) \Big] \cdots \tau_{t_{n}}(\mathcal{P}(\eta_{\beta} t_{m})) \Big] \Big\rangle_{\beta,L} \\
&\quad = \sum_{n_{1} \geq 1}\cdots \sum_{n_{m}\geq 1} \frac{e^{n_{1} \eta_{\beta} t_{1}}}{n_{1}!} \cdots \frac{e^{n_{m} \eta_{\beta} t_{m}}}{n_{m}!}\Big\langle \Big[\cdots \Big[ \Big[ \mathcal{O}_{z}, \tau_{t_{1}}(\mathcal{P}^{A}_{n_{1}}) \Big], \tau_{t_{2}}(\mathcal{P}^{A}_{n_{2}}) \Big] \cdots \tau_{t_{n}}(\mathcal{P}^{A}_{n_{m}}) \Big] \Big\rangle_{\beta,L}\;,
\end{split}
\end{equation}
where $\mathcal{P}^{A}_{n_{1}}$ is of order $n_{1}$ in $A$, recall (\ref{eq:PnA}). In particular, (\ref{eq:duha}) can be rewritten as:
\begin{equation}\label{eq:lth}
\begin{split}
&\Tr_{\mathcal{F}} \mathcal{O}_{z} \widetilde{\Gamma}(0) - \Tr_{\mathcal{F}} \mathcal{O}_{z} \Gamma_{\beta,L} \\
&\qquad = \sum_{\ell \geq 1} \sum_{m = 1}^{\ell} (-i)^{m} \sum_{\substack{\{n_{i}\}_{i=1}^{m}: n_{i}\geq 1, \\ n_{1} + \ldots + n_{m} = \ell}} \int_{-\infty \leq t_{m} \leq t_{m-1} \leq \ldots \leq t_{1} \leq 0} d\underline{t} \,\frac{e^{n_{1} \eta_{\beta} t_{1}}}{n_{1}!} \cdots \frac{e^{n_{m} \eta_{\beta} t_{m}}}{n_{m}!} \\
&\quad\qquad \cdot \Big\langle \Big[\cdots \Big[ \Big[ \mathcal{O}_{z}, \tau_{t_{1}}(\mathcal{P}^{A}_{n_{1}}) \Big], \tau_{t_{2}}(\mathcal{P}^{A}_{n_{2}}) \Big] \cdots \tau_{t_{n}}(\mathcal{P}^{A}_{n_{m}}) \Big] \Big\rangle_{\beta,L} \\
&\qquad\equiv \sum_{\ell \geq 1} \Big(\Tr_{\mathcal{F}} \mathcal{O}_{z} \widetilde{\Gamma}(0)\Big)^{(\ell)}
\end{split}
\end{equation}
where $(\cdot)^{(\ell)}$ denotes the $\ell$-th order in $A$. Following \cite{GLMP24}, Lemma 4.2, we use that the real-time integral of the right-hand side of (\ref{eq:lth}) can be conveniently rewritten in imaginary time. We have:
\begin{equation}\label{eq:wick}
\begin{split}
&\sum_{\substack{\{n_{i}\}_{i=1}^{m}: n_{i}\geq 1, \\ n_{1} + \ldots + n_{m} = \ell}}\int_{-\infty \leq t_{m} \ \leq \ldots \leq t_{1} \leq 0} d\underline{t} \,  \frac{e^{n_{1} \eta_{\beta} t_{1}}}{n_{1}!} \cdots \frac{e^{n_{m} \eta_{\beta} t_{m}}}{n_{m}!}\Big\langle \Big[\cdots \Big[ \Big[ \mathcal{O}_{z}, \tau_{t_{1}}(\mathcal{P}^{A}_{n_{1}}) \Big], \tau_{t_{2}}(\mathcal{P}^{A}_{n_{2}}) \Big] \cdots \tau_{t_{n}}(\mathcal{P}^{A}_{n_{m}}) \Big] \Big\rangle_{\beta,L} \\
& \quad = \frac{(-i)^{m}}{m!} \sum_{\substack{\{n_{i}\}_{i=1}^{m}: n_{i}\geq 1, \\ n_{1} + \ldots + n_{m} = \ell}}\int_{[0,\beta]^{m}} d\underline{t}\, \frac{e^{i n_{1} \eta_{\beta} t_{1}}}{n_{1}!} \cdots \frac{e^{i n_{m} \eta_{\beta} t_{m}}}{n_{m}!} \langle {\bf T} \gamma_{t_{1}}(\mathcal{P}^{A}_{n_{1}}) \;; \gamma_{t_{2}}(\mathcal{P}^{A}_{n_{2}})\;; \cdots \;; \gamma_{t_{m}}(\mathcal{P}^{A}_{n_{m}})\;; \mathcal{O}_{z}  \rangle_{\beta,L}\;.
\end{split}
\end{equation}
The reason for approximating $\eta$ with $\eta_{\beta}$ is this identity, which only holds for these special values of the adiabatic parameter (``bosonic Matsubara frequencies''). Let us explain the meaning of the objects involved in the right-hand side. The operator $\gamma_{s}(\cdot)$ is the Euclidean evolution,
\begin{equation}
\gamma_{s}(\mathcal{A}) = e^{s(\mathcal{H} - \mu \mathcal{N})} \mathcal{A} e^{-s(\mathcal{H} - \mu \mathcal{N})}\;,\qquad 0\leq s\leq \beta\;.
\end{equation}
The argument of the integral in (\ref{eq:wick}) is a time-ordered cumulant, see \cite{GLMP24}, Eq. (2.14), for a definition in terms of the usual correlations. The symbol ${\bf T}$ denotes the fermionic time-ordering. It acts as follows on any fermionic monomial as, omitting the $\rho$-labels:
\begin{equation}
{\bf T} \gamma_{t_{1}}(a^{\sharp_{1}}_{x_{1}}) \cdots \gamma_{t_{n}}(a^{\sharp_{n}}_{x_{n}}) = (-1)^{\pi} \gamma_{t_{\pi(1)}}(a^{\sharp_{\pi(1)}}_{x_{\pi(1)}}) \cdots \gamma_{t_{\pi(n)}}(a^{\sharp_{\pi(n)}}_{x_{\pi(n)}})
\end{equation} 
for all $t_{i}$ in $[0,\beta)$, with $a^{\sharp}$ equal to $a$ or $a^{*}$, and where $\pi$ is the permutation such that $t_{\pi(1)} > t_{\pi(2)} > \cdots > t_{\pi(n)}$. In case two times are equal, the ambiguity is solved by normal order. Since the Hamiltonian is invariant under translations, the Gibbs state is invariant under translations as well. Also, the time-ordered cumulants, initially defined only for $t_{i} \in [0,\beta)$, can be periodically extended to all $t_{i}\in \mathbb{R}$, ultimately thanks to the Kubo-Martin-Schwinger identity. The resulting function is periodic with period $\beta$ in all times, it is continuous at $t_{i} \in \beta\mathbb{Z}$ and it is time-translation invariant. See \cite{GLMP24} for a review of these known facts.

In our case, the time-ordered cumulant in (\ref{eq:wick}) is evaluated for a quasi-free state, and therefore it can be computed in terms of the two-point function, 
\begin{equation}
g_{\rho_{1}\rho_{2}}((t_{1},x_{1}); (t_{2}, x_{2})) := \langle {\bf T} \gamma_{t_{1}}(a_{x_{1},\rho_{1}}) \gamma_{t_{2}}(a^{*}_{x_{2},\rho_{2}}) \rangle_{\beta,L}
\end{equation}
using the fermionic Wick's rule. Given the form (\ref{eq:PnA}) of the operators $\mathcal{P}^{A}_{n}$, the outcome is particularly simple: it is given by a sum of loop Feynman diagrams. To introduce them, let us represent the two-point function in momentum space, as follows:
\begin{equation}\label{eq:2pt2}
\begin{split}
g((t_{1},x_{1}); (t_{2}, x_{2})) &= \lim_{N\to \infty} \frac{1}{\beta L^{2}}  \sum_{\substack{k_{0} \in M_{\beta} \\ |k_{0}| \leq N}} \sum_{k \in \mathcal{C}_{1,L}^{*}} e^{ik_{0} (t_{1} - t_{2})} e^{i k\cdot (x_{1} - x_{2})} \frac{1}{ik_{0} + H(k) - \mu} \\
&\equiv \lim_{N\to \infty} \frac{1}{\beta L^{2}}  \sum_{\substack{k_{0} \in M_{\beta} \\ |k_{0}| \leq N}}\sum_{k \in \mathcal{C}_{1,L}^{*}} e^{ik_{0} (t_{1} - t_{2})} e^{i k\cdot (x_{1} - x_{2})} \hat g({\bf k})
\end{split}
\end{equation}
where: $M_{\beta} = \frac{2\pi}{\beta} (\mathbb{Z} + 1/2)$ is the set of fermionic Matsubara frequencies; ${\bf k} = (k_{0}, k) \in M_{\beta} \times \FC_{1,L}^{*}$ with $\FC_{1,L}^{*}$ as in (\ref{eq:CL}); and $\hat g({\bf k})^{-1} = ik_{0} + H(k) - \mu$. The identity (\ref{eq:2pt2}) holds at non-coinciding space-time points, and it is all we will need to evaluate the cumulants in (\ref{eq:wick}). For $\underline{\alpha} = (\alpha_{1}, \ldots, \alpha_{n})$, and $\underline{p} = (p_{1}, \ldots, p_{n})$, let:
\begin{equation}
\hat A^{(n)}_{\theta,\underline{\alpha}}(\underline{p}) := \prod_{i=1}^{n} \hat A_{\theta,\alpha_{1}}(p_{1}) \cdots \hat A_{\theta,\alpha_{n}}(p_{n})
\end{equation}
where $\hat A_{\theta,\alpha}(p)$ is defined in (\ref{eq:hatA}). We rewrite (\ref{eq:PnA}) in a more compact form as, setting $P = \sum_{i=1}^{n} p_{i}$:
\begin{equation}
\begin{split}
\mathcal{P}^{A}_{n} &= \frac{1}{L^{2}} \sum_{k} \frac{1}{L^{2n}} \sum_{\underline{p}} \sum_{\underline{\alpha}} \sum_{\rho,\rho'} \hat A^{(n)}_{\theta,\underline{\alpha}}(\underline{p}) \hat a^{*}_{k + P,\rho} \hat a_{k,\rho'} \hat Q_{n;\underline{\alpha}, \rho, \rho'}(k,\underline{p}) \\
&\equiv \frac{1}{L^{2n}} \sum_{\underline{p}} \sum_{\underline{\alpha}}\hat A^{(n)}_{\theta,\underline{\alpha}}(\underline{p})  \hat W_{n,\underline{\alpha}}(\underline{p})\;.
\end{split}
\end{equation}
%
%
%
%
%
Let:
\begin{equation}
\hat W_{n,\underline{\alpha}}(n\eta_{\beta}; \underline{p}) := \int_{0}^{\beta} dt\, e^{i n \eta_{\beta} t} \gamma_{t}(\hat W_{n,\underline{\alpha}}(\underline{p}))\;.
\end{equation}
Thus, we have:
\begin{equation}
\begin{split}
&\int_{[0,\beta]^{m}} d\underline{t}\, e^{i n_{1} \eta_{\beta} t_{1}} \cdots e^{i n_{m} \eta_{\beta} t_{m}} \langle {\bf T} \gamma_{t_{1}}(\mathcal{P}^{A}_{n_{1}}) \;; \gamma_{t_{2}}(\mathcal{P}^{A}_{n_{2}})\;; \cdots \;; \gamma_{t_{m}}(\mathcal{P}^{A}_{n_{m}})\;; \mathcal{O}_{z}  \rangle_{\beta,L} \\
&=\frac{1}{L^{2(n_{1} + \ldots + n_{m})}} \sum_{\{\underline{p}_{i}\},\, \{\underline{\alpha}_{i}\}}\hat A^{(n_{1})}_{\theta,\underline{\alpha}_{1}}(\underline{p}_{1}) \cdots  \hat A^{(n_{m})}_{\theta,\underline{\alpha}_{m}}(\underline{p}_{m})\cdot\\
&\qquad\cdot \langle {\bf T} \hat W_{n_{1},\underline{\alpha}_{1}}(n_{1}\eta_{\beta}; \underline{p}_{1})\;; \cdots\;; \hat W_{n_{m},\underline{\alpha}_{m}}(n_{m}\eta_{\beta}; \underline{p}_{m})\;; \mathcal{O}_{z} \rangle_{\beta,L}\;;
\end{split}
\end{equation}
by space-time translation invariance of the time-ordered correlation functions:
\begin{equation}\label{eq:momspace}
\begin{split}
&\langle {\bf T} \hat W_{n_{1},\underline{\alpha}_{1}}(n_{1}\eta_{\beta}; \underline{p}_{1})\;; \cdots\;; \hat W_{n_{m},\underline{\alpha}_{m}}(n_{m}\eta_{\beta}; \underline{p}_{m})\;; \mathcal{O}_{z} \rangle_{\beta,L} \\
&\quad = \frac{e^{iz\cdot (P_{1} + \ldots +P_{m})}}{\beta L^{2}}\langle {\bf T} \hat W_{n_{1},\underline{\alpha}_{1}}(n_{1}\eta_{\beta}; \underline{p}_{1})\;; \cdots\;; \hat W_{n_{m},\underline{\alpha}_{m}}(n_{m}\eta_{\beta}; \underline{p}_{m})\;; \hat{\mathcal{O}}_{-{\bf P}_{1} - \ldots - {\bf P}_{m}} \rangle_{\beta,L}
\end{split}
\end{equation}
with ${\bf P}_{i} = (n_{i}\eta_{\beta}, P_{i})$, $P_{i} = \sum_{j = 1}^{n} p_{i,j}$, and:
\begin{equation}
\hat{\mathcal{O}}_{-{\bf P}_{1} - \ldots - {\bf P}_{m}} := \int_{0}^{\beta} dt\, e^{-i \sum_{j=1}^{m} n_{j} \eta_{\beta} t} \sum_{z\in \Lambda_{L}} e^{-i \sum_{j=1}^{m}P_{j} \cdot z} \gamma_{t}(\mathcal{O}_{z})\;.
\end{equation}
We are now ready to apply Wick's theorem to evaluate the cumulant. We have, omitting labels to avoid cluttering the notation:
\begin{equation}
\label{eq:props}
\begin{split}
&(\ref{eq:momspace}) = -  \frac{1}{\beta L^{2(1 + n_{1} + \ldots + n_{m})}}\cdot\\
&\qquad\qquad\cdot \sum_{\pi \in S_{m+1}} \sum_{{\bf k}} \sum_{\underline{p}_{1}, \ldots, \underline{p}_{m}} \hat A^{(n_{\pi(1)})}_{\theta}(\underline{p}_{\pi(1)})\cdots \hat A^{(n_{\pi(m)})}_{\theta}(\underline{p}_{\pi(m)})e^{iz\cdot (P_{1} + \ldots +P_{m})}\\
&\qquad\qquad\cdot \tr \Big\{ \hat g({\bf k}) Q_{\pi(1)} \hat g({\bf k} + {\bf P}_{\pi(1)})  Q_{\pi(2)}  \hat g({\bf k} + {\bf P}_{\pi(1)} + {\bf P}_{\pi(2)}) \\
&\qquad\qquad\quad\cdots Q_{\pi(m)} \hat g({\bf k} + {\bf P}_{\pi(1)} + \ldots + {\bf P}_{\pi(m)}) Q_{\pi(m+1)}\Big\}\;,
\end{split}
\end{equation}
where $\pi$ is a permutation of the natural numbers $\{1,2,\ldots, m+1\}$, and with the understanding that $\hat Q_{m+1} = \hat O$.  Now, since by the conical intersections of the energy bands we have, for $|k - \omega_{l}|$ small enough, $ |\Lambda_{\pm}(k) - \mu| \geq C |k-\omega_{l}|$, the following estimate for the two-point function in momentum space holds:
\begin{equation}\label{eq:gest0}
\| \hat g({\bf k}) \| \leq \sum_{l} \frac{C}{|k_{0}| + |k - \omega_{l}|}
\end{equation}
where the sum runs over the Fermi points. Let us discuss the estimate for the ${\bf k}$ sum in (\ref{eq:props}). For a given ${\bf k}$, suppose that all the arguments of the propagators are in norm bigger than $\eta_{\beta} /2$. Then, every propagator is bounded by $C/\eta_{\beta}$; we use this bound for all propagators except three, to be integrated. Alternatively, suppose that one of the arguments of the propagators is in norm less than $\eta_{\beta}/2$. Then, the arguments of the other propagators are in norm at least $\eta_{\beta}/2$, due to the shift of at least $\eta_{\beta}$ in the temporal component of the other momenta. In this second case, we again estimate by $C/\eta_{\beta}$ all propagators except three, and we include the propagator with the smallest denominator within the three to be integrated. All in all we have, for $\beta, L$ large enough:
\begin{equation}
\begin{split}
\frac{1}{\beta L^{2}} &\sum_{{\bf k}} \Big| \tr \Big\{ \hat g({\bf k}) Q_{\pi(1)} \hat g({\bf k} + {\bf P}_{\pi(1)})  Q_{\pi(2)}  \hat g({\bf k} + {\bf P}_{\pi(1)} + {\bf P}_{\pi(2)})\\& \cdots Q_{\pi(m)} \hat g({\bf k} + {\bf P}_{\pi(1)} + \ldots + {\bf P}_{\pi(m)}) Q_{\pi(m+1)}\Big\} \Big| \\
& \leq \frac{C^{m}}{\eta^{m-2}} \sup^{*}_{n_{a}, n_{b}, n_{c}} \sup_{P_{a}, P_{b}, P_{c}}  \int_{\mathbb{R} \times \mathbb{T}^{2}} d{\bf k}\, \frac{1}{|k_{0} + n_{a}\eta| + |k + P_{a}|} \frac{1}{|k_{0} + n_{b}\eta| + |k + P_{b}|} \\&\qquad\qquad \cdot \frac{1}{|k_{0} + n_{c}\eta| + |k + P_{c}|}
\end{split}
\end{equation}
where the star on the supremum means that $n_{a}$, $n_{b}$, $n_{c}$ are three different natural numbers. Let ${\bf P}_{a} = (n_{a}\eta, P_{a})$. Denoting by $\theta(A)$ the function which is equal to $1$ if the condition $A$ holds and $0$ otherwise, we write:
\begin{equation}\label{eq:int}
\begin{split}
&\int_{\mathbb{R} \times \mathbb{T}^{2}} d{\bf k}\, \frac{1}{|k_{0} + n_{a}\eta| + |k + P_{a}|} \frac{1}{|k_{0} + n_{b}\eta| + |k + P_{b}|}\frac{1}{|k_{0} + n_{c}\eta| + |k + P_{c}|} \\
&\quad = \int_{\mathbb{R} \times \mathbb{T}^{2}} d{\bf k}\, \frac{\theta(\| {\bf k} + {\bf P}_{a}\| > \eta/2)}{|k_{0} + n_{a}\eta| + |k + P_{a}|} \frac{1}{|k_{0} + n_{b}\eta| + |k + P_{b}|}\frac{1}{|k_{0} + n_{c}\eta| + |k + P_{c}|}\\
&\qquad + \int_{\mathbb{R} \times \mathbb{T}^{2}} d{\bf k}\, \frac{\theta(\| {\bf k} + {\bf P}_{a}\| \leq \eta/2)}{|k_{0} + n_{a}\eta| + |k + P_{a}|} \frac{1}{|k_{0} + n_{b}\eta| + |k + P_{b}|}\frac{1}{|k_{0} + n_{c}\eta| + |k + P_{c}|} \equiv \text{I} + \text{II}\;.
\end{split}
\end{equation}
Consider $\text{II}$. The constraint implies that $|k_{0} + n_{a}\eta| \leq \eta/2$. Since the natural numbers $n_{a}, n_{b}, n_{c}$ are different, we have that $|k_{0} + n_{b}\eta| \geq \eta/2$, $|k_{0} + n_{c}\eta| \geq \eta/2$. Then:
\begin{equation}
\text{II} \leq \frac{C}{\eta^{2}} \int_{\mathbb{R} \times \mathbb{T}^{2}} d{\bf k}\, \frac{\theta(\| {\bf k} + {\bf P}_{a}\| \leq \eta/2)}{|k_{0} + n_{a}\eta| + |k + P_{a}|} \leq C\;.
\end{equation}
Repeating a similar argument in $\text{I}$, we easily get:
\begin{equation}
\begin{split}
(\ref{eq:int}) &\leq K +  \int_{\mathbb{R} \times \mathbb{T}^{2}} d{\bf k}\, \frac{\theta(\| {\bf k} + {\bf P}_{a}\| > \eta/2)}{|k_{0} + n_{a}\eta| + |k + P_{a}|} \frac{\theta(\| {\bf k} + {\bf P}_{b}\| > \eta/2)}{|k_{0} + n_{b}\eta| + |k + P_{b}|}\frac{\theta(\| {\bf k} + {\bf P}_{c}\| > \eta/2)}{|k_{0} + n_{c}\eta| + |k + P_{c}|} \\
&\leq K + C\int_{\mathbb{R} \times \mathbb{T}^{2}} d{\bf k}\, \frac{\theta(\| {\bf k}\| > \eta/2)}{\| {\bf k}\|^{3}} \\
&\leq K + C|\log \eta|\;.
\end{split}
\end{equation}
Therefore:
\begin{equation}\label{eq:loops}
\begin{split}
&\frac{1}{\beta L^{2}} \sum_{{\bf k}} \Big| \tr \Big\{ \hat g({\bf k}) Q_{\pi(1)} \\&\quad\qquad \cdot \hat g({\bf k} + {\bf P}_{\pi(1)})  Q_{\pi(2)}  \hat g({\bf k} + {\bf P}_{\pi(1)} + {\bf P}_{\pi(2)}) \cdots Q_{\pi(n-1)} \hat g({\bf k} + {\bf P}_{\pi(1)} + \ldots + {\bf P}_{\pi(m)}) Q_{\pi(m)}\Big\} \Big| \\
&\qquad \leq \frac{C^{m}}{\eta^{m-2}}|\log \eta|\;.
\end{split}
\end{equation}
Next, to control the sum over $\underline{p}_{i}$ in (\ref{eq:props}) we use that, by (\ref{eq:Aests}):
\begin{equation}\label{eq:Aest}
\frac{1}{L^{2n_{i}}} \sum_{\underline{p}_{i}} \Big| \hat A^{(n_{i})}_{\theta}(\underline{p}_{i}) \Big|  \leq C^{n_{i}} \theta^{n_{i}}\;.
\end{equation}
All in all, from (\ref{eq:props}), (\ref{eq:loops}), (\ref{eq:Aest}), we have:
\begin{equation}
\begin{split}
|(\ref{eq:wick})| &\leq  \sum_{\substack{\{n_{i}\}_{i=1}^{m}: n_{i}\geq 1, \\ n_{1} + \ldots + n_{m} = \ell}} \frac{1}{n_{1}!\cdots n_{m}!}\frac{(C \theta)^{\sum_{i} n_{i}}}{\eta^{m}} \eta^{2} |\log \eta| \\
&\leq \sum_{\substack{\{n_{i}\}_{i=1}^{m}: n_{i}\geq 1, \\ n_{1} + \ldots + n_{m} = \ell}} {\ell \choose n_{1} \ldots n_{m}} \frac{1}{\ell!}\frac{(C \theta)^{\ell}}{\eta^{m}} \eta^{2} |\log \eta| \\
&\leq \frac{m^{\ell}}{\ell!} \frac{(C \theta)^{\ell}}{\eta^{m}} \eta^{2} |\log \eta|\;.
\end{split}
\end{equation}
Observe that the factor $1/m!$ in (\ref{eq:wick}) has been used to control the sum over the permutations. Thus, recalling (\ref{eq:lth}), we have:
\begin{equation}
\begin{split}
\Big| \Big(\Tr_{\mathcal{F}} \mathcal{O}_{z} \widetilde{\Gamma}(0)\Big)^{(\ell)} \Big| &\leq \sum_{m = 1}^{\ell} \frac{m^{\ell}}{\ell!} \frac{(C \theta)^{\ell}}{\eta^{m}} \eta^{2} |\log \eta| \\
&\leq K^{\ell} \Big(\frac{\theta}{\eta}\Big)^{\ell} \eta^{2} |\log \eta|\;.
\end{split}
\end{equation}
In particular, this bound shows that, for $ |\theta|/\eta$ sufficiently small: 
\begin{equation}\label{eq:duharest}
\Big|\sum_{\ell \geq 2} \Big(\Tr_{\mathcal{F}} \mathcal{O}_{z} \widetilde{\Gamma}(0)\Big)^{(\ell)} \Big| \leq C \theta^{2}  |\log \eta|\;.
\end{equation}
This concludes the proof of Theorem \ref{thm:response}.\qed
\begin{remark} Let us comment about the bound for the first order term. In this case, the expression (\ref{eq:props}) only involves two propagators, and it is integrable in ${\bf k}$ uniformly in $\eta$. Combined with $\| \hat A_{\theta} \|_{\ell^{1}} \leq C|\theta|$, this shows that the first order term is bounded by $\text{(const.)} |\theta|$, as claimed in Remark \ref{rem:resp}.
\end{remark}

%
%
%
%
%
%

\paragraph{Acknowledgments} G.~M. and M.~P. acknowledge support by the European Research Council through the ERC-StG MaMBoQ, n. 802901. G.~M. acknowledges financial support from the Independent Research Fund Denmark--Natural Sciences, grant DFF–10.46540/2032-00005B and from the European Research Council through the ERC CoG UniCoSM, grant agreement n.724939. M.~P. acknowledges support from the MUR, PRIN 2022 project MaIQuFi cod. 20223J85K3. This work has been carried out under the auspices of the GNFM of INdAM. We thank the anonymous referees for comments on a previous version of this manuscript.

\vspace{0.5cm}

\noindent\textbf{Statements and Declarations}\\ 
\textbf{Competing Interests} All authors declare that they have no conflicts of interest to disclose.

\vspace{0.5cm}

\noindent\textbf{Data Availability} Data sharing is not applicable to this article as no datasets were generated or
analyzed during the current study.

\appendix
\section{Auxiliary results}
\label{sec:aux}

In this appendix we collect the proofs of some auxiliary results.

\begin{proof}[Proof of Lemma \ref{lem:sufficient-longitudinal-quantization}]

Let us recall that 
\begin{equation}
\label{eqn:det}
        \det S = \sqrt{s_{11}^2 + s_{21}^2}\sqrt{s_{12}^2 + s_{22}^2}\sin\theta,
\end{equation}
where $\theta$ is the angle between the vectors $(s_{11},s_{21}),(s_{12}, s_{22})$. If the following conditions hold: $s_{11}^2 + s_{21}^2 = s_{12}^2 + s_{22}^2$ and $(s_{11},s_{21})\cdot(s_{12}, s_{22})=0$, then $|\det S|= s_{11}^2 + s_{21}^2=s_{12}^2 + s_{22}^2$ and so equality \eqref{eqn:condS} is implied. Viceversa, if \eqref{eqn:condS} holds true, then $s_{11}^2 + s_{21}^2 =\abs{\det S} =s_{12}^2 + s_{22}^2$; thus, we get that $\abs{\det S}=\sqrt{s_{12}^2 + s_{22}^2}\sqrt{s_{12}^2 + s_{22}^2}$, and so by comparison with \eqref{eqn:det} we have that $\cos \theta=0$. 
\end{proof}

\begin{proof}[Proof of Lemma \ref{lem:norm-derivative-projector-estimate}]
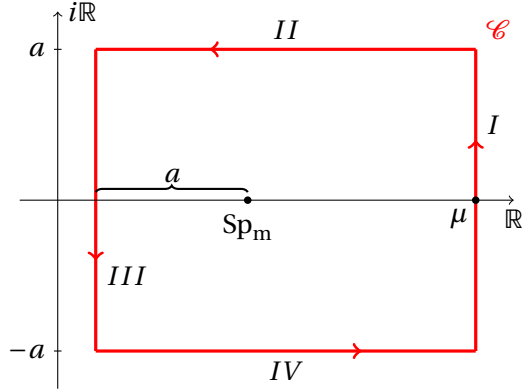
\begin{figure}[h]
    \centering
    \begin{tikzpicture}
        \begin{scope}[very thick,decoration={
            markings,
            mark=at position 0.7 with {\arrow{>}}}
            ]
            \draw[postaction={decorate},color=red] (6,-2) -- node[anchor=west,yshift=1cm] {{\color{black} $I$}} ++(0,4) node[anchor=south west] {$\mathcal{C}$};
            \draw[postaction={decorate},color=red] (6,2)-- node[anchor=south] {{\color{black} $II$}} ++(-5,0);
            \draw[postaction={decorate},color=red] (1,2)-- node[anchor=west,yshift=-1cm] {{\color{black} $III$}} ++(0,-4);
            \draw[postaction={decorate},color=red] (1,-2)-- node[anchor=north] {{\color{black} $IV$}} ++(5,0);
        \end{scope}
        \draw[decorate, decoration={brace}, yshift=3pt, thick] (1,0) -- node[above] {$a$} (3,0);
        \draw[->] (0,0) -- (6.5,0) node[below] {$\R$};
        \draw[->] (0.5,-2.5) -- ++(0,5) node[anchor=west] {$i\R$};
        \path (6,0) node[circle,fill=black, inner sep=1pt] {} -- ++(0,0) node[anchor=north east, inner sep=3pt] {$\mu$};
        \path (3,0) node[circle,fill=black, inner sep=1pt] {} -- ++(0,0) node[anchor=north, inner sep=5pt] {$\footnotesize \mathrm{Sp}_\text{m}$};
        \foreach \y/\name in {-2/-a, 2/a}
            \draw (0.55,\y cm) -- (0.45,\y cm) node[anchor=east,fill=white] {$\name$};
    \end{tikzpicture}
    \caption{Complex path used to prove Lemma \ref{lem:norm-derivative-projector-estimate}. We denote the minimum of the spectrum of $H$ as $\mathrm{Sp}_\text{m}$.}
    \label{fig:path-projector}
\end{figure}
For any $k\in \FC_1^*\setminus\{\omega_1,\dots,\omega_n\}$ we can choose a specific complex path $\mathcal{C}$, as illustrated in Figure \ref{fig:path-projector}, by setting $\mathrm{Sp}_{\text{m}}:=\min\{\mathrm{Sp}(H(k))\}$, to define the Fermi projector $P_\mu(k)$  via the Riesz formula:
 \begin{equation}
 \label{eqn:fermiprojriesz}
 P_\mu(k)=\frac{1}{2\pi i}\oint_{\mathcal{C}}dz\, (z \Id-H(k))^{-1}.
\end{equation}
The fact that $\mathcal{C}$ does not depend on $k$ implies that $P_\mu(k)$ is differentiable at every $k$ which is not a Fermi point.
Moreover, for every $1\leq j\leq 2$ we have that
\begin{equation}
    \label{eqn:projector-derivative-cauchy}
    \partial_j P_\mu(k) = \frac{1}{2\pi i}\int_{\mathcal{C}}dz\, (z\Id-H(k))^{-1}\,\partial_j H(k)\,(z\Id-H(k))^{-1}.
\end{equation}
We shall estimate $\norm{\partial_j P_\mu(k)}$ by using the particular choice for the closed curve $\mathcal{C}$ and expression \eqref{eqn:projector-derivative-cauchy}.
Let us define $M:=\max_{k\in \FC_1^*}\norm{\partial_j H(k)}$.
Notice that for the contribution coming from the path $II$ we get that
\begin{align*}
&\norm{\int_{II}dz\, (z\Id-H(k))^{-1}\,\partial_j H(k)\,(z\Id-H(k))^{-1}}\leq M \int_{b-a}^{\mu}dx\, \norm{(x\Id-H(k))^{-1}}^2\leq  M\frac{\mu -b +a}{a^2},
\end{align*}
where we have used that $\norm{(x\Id-H(k))^{-1}}=\dist (\mathrm{Sp}(H(k)), x)^{-1}\leq \frac{1}{a}$. Similarly we can estimate the contributions coming from paths $III, IV$. While for the term coming from $I$ we obtain that
\begin{align*}
&\norm{\int_{I} dz\, (z\Id-H(k))^{-1}\,\partial_j H(k)\,(z\Id-H(k))^{-1}}\leq 
\frac{M\pi}{\mu - \Lambda_-(k)}\leq C\sum_{l=1}^n\frac{1}{\abs{k-\omega_l}},
\end{align*}
where we have used equality \eqref{eqn:conical} and inequality \eqref{eqn:cone}. 
\end{proof}

\begin{proof}[Proof of Proposition \ref{prp:hjj-two-eigenvalues}]
First, by using Remark \ref{rmk:complex-deformation-k-domain} we undo the complex deformation in time of $\fjjsing(\eta)$ on the integration domain $B_\eps$, rewriting it in real time. Thus, by also exploiting rewriting \eqref{eqn:rewf}, for any $\eta>0$ we have that
\begin{equation*}
\fjjsing(\eta) = \frac{i}{(2\pi)^2} \int_{B_\eps} dk \int_{-\infty}^0dt e^{\eta t} \Tr\left(P_b(k)\left[ \tau_t((\chi_<J_j\chi_<)(k)), (\chi_<J_j\chi_<)(k)\right] P_b(k) \right).
\end{equation*}
By using that $\chi_<(k)=P_-(k)+P_+(k)$ for every $k\in B_\eps$ due to condition \eqref{eqn:condition-eps-delta}, we get that (here, we do not write the $k$-dependence for the sake of readability) 
\begin{align*}
P_b\left[ \tau_t(\chi_<J_j\chi_<), \chi_<J_j\chi_<\right] P_b
&=e^{i (\Lambda_- -\Lambda_+) t} P_- J_j P_+J_j P_- - e^{i (\Lambda_+-\Lambda_-) t}P_- J_j P_+J_j P_-,
\end{align*}
where we have exploited that $P_b(k)P_+(k)=0$ and $P_b(k)P_-(k)=P_-(k)$ \alev in $k$.
Therefore, we conclude that
\begin{align*}
\fjjsing(\eta) 
&=\frac{2}{(2\pi)^2}\int_{B_\eps} dk\frac{(\Lambda_-(k) - \Lambda_+(k))}{\eta^2 + (\Lambda_-(k) - \Lambda_{+}(k))^2} \Tr\big(  P_-(k)J_j(k)P_+(k)J_j(k)\big).
\end{align*}
\end{proof}

\begin{proof}[Proof of Proposition \ref{prp:longitudinal-conductivity-eigenvalues} ]
We start by defining the interval $I_e := (-\infty, \mu-2\delta)\cup(\mu+2\delta,\infty)$ and its characteristic function $\chi_{I_e}$.
Then we introduce the following fibered operators:
\begin{align}
\label{eqn:tildeH-definition}
    \tilde{H}(k) &:= (\Lambda_-(k)-\mu) P_-(k) + (\Lambda_+(k)-\mu) P_+(k)\equiv( H(k)-\mu)(P_-(k)+P_+(k)), \\
    P_e(k) &:= \chi_{I_e}(H(k)).
\end{align}
By \eqref{eqn:condition-eps-delta} we get that $k\in B_\eps$ implies that the eigenvalues $\Lambda_\pm(k)\in (\mu-\delta,\mu+\delta)$ and the other eigenvalues of $H(k)$ are in the interval $I_e$.
Hence, we can write the identity $\Id$ on $\ell^2(\X/\Gamma)$ as
\begin{equation}
\label{eqn:decid}
    \Id = P_-(k) + P_+(k) + P_e(k) \qquad \text{for every } k \in B_\eps.
\end{equation}
Thus, by \eqref{eqn:sumPpm} we obtain that the map $B_\eps\ni k\mapsto P_e(k)$ is $C^2$ and by Assumption \ref{assum:H}\ref{item:perselfadj} the function $B_\eps\ni k\mapsto \tilde{H}(k)$ is $C^2$ as well.

Now we shall rewrite $\fjjsing$ as the sum of the function $\zjj$ defined in \eqref{eqn:zjj-definition} and the map $\xi$ defined in \eqref{eqn:xi}, which does not contribute to $\sigma_{jj}$, as shown below.
In the following we omit the explicit dependence on the $k$-variable and the equalities have to be understood \alev in $k\in B_\eps$ (actually, we shall only exclude the Fermi points $\omega_1,\dots,\omega_n$ in $B_\eps$). We also use the shorthand notation:
\begin{equation*}
    \Jj^{+-} := P_+(\partial_j\Tilde{H})P_-, \qquad \Jj^{-+} := P_-(\partial_j\Tilde{H})P_+.
\end{equation*}
In view of the identity $J_j= \partial_j H = \partial_j(H-\mu)=\partial_j(\Tilde{H} +(H-\mu)P_e)$, we rewrite the trace term inside the integral in \eqref{eqn:fjjsing-two-eigenvalues-current} as
\begin{equation}
\label{eqn:rewriting-trace-fjjsing}
    \Tr\big( P_- \,J_j\, P_+\, J_j\, P_-   \big)= 
    \Tr\left( \Jj^{-+}\, \Jj^{+-} \right)
\end{equation}
where we have used that $P_e\,P_+=0$, $(\partial_j P_e)P_+ = - P_e(\partial_j P_+)$ together with $P_-\, P_e=0$.
Now we rewrite the expression of $\Tr\big(\Jj^{-+}\Jj^{+-}\big)$, by decoupling the contributions containing only $P_-$ and the ones involving only $P_+$:
\begin{equation}
\label{eqn:rewrite-JjTilde-1}
\begin{split}
    \Tr\big(\Jj^{-+}\Jj^{+-}\big)  
    &= \frac{1}{2}\Tr\big(P_-(\partial_j \tilde{H})^2\big) - \frac{1}{2}\Tr\big(P_-(\partial_j \tilde{H})P_-(\partial_j\tilde{H})P_-) \\ 
    & \quad - \frac{1}{2}\Tr\big(P_-(\partial_j \tilde{H})P_e(\partial_j \tilde{H})P_-\big) + (m \leftrightarrow m+1),
\end{split}
\end{equation}
where we have used \eqref{eqn:decid} and the cyclicity of the trace, and the notation \virg{$+(m \leftrightarrow m+1)$} means that it is added the term given by substituting $m$ with $m+1$ to all the summands to the left.
To extract the function $\zjj$ from $\fjjsing$, we proceed by computing $\partial_j \tilde{H}$. By the very definition of $\tilde{H}$ in \eqref{eqn:tildeH-definition}, we have that
\begin{align*}
\partial_j \tilde{H}
=(\partial_j\Lambda_-) P_-  + (\Lambda_- -\mu)(\partial_j P_- )+ (\partial_j\Lambda_+) P_+ +(\Lambda_+ -\mu)(\partial_j P_+).
\end{align*}
Then, we obtain:
\begin{align}
\Tr\big(\Jj^{-+}\Jj^{+-}\big)  
\label{eqn:rewJtilde}
&=\frac{1}{2}\Tr\big((\partial_j \tilde{H})^2\big)-\Tr\big(P_e(\partial_j \tilde{H})^2P_e\big)- \frac{1}{2} {(\partial_j\Lambda_-)}^2 - \frac{1}{2} {(\partial_j\Lambda_+)}^2.
\end{align}
Next, we proceed by rewriting the first summand in \eqref{eqn:rewJtilde} as
\begin{align}
\label{eqn:trdeltatildeH2}
\begin{split}
    \Tr\big((\partial_j \tilde{H})^2\big) & = \partial_j\Tr\big( \tilde{H}\,\partial_j \tilde{H} \big) - \Tr\big(\tilde{H}\,\partial_j^2\tilde{H}\big) 
    = \frac{1}{2}\partial_j^2\Tr\big(\tilde{H}^2\big) - \Tr\big(\tilde{H}\,\partial_j^2 \tilde{H} \big) \\
    &= \frac{1}{2}\partial_j^2\big((\Lambda_--\mu)^2 + (\Lambda_+-\mu)^2\big) - \Tr\big(\tilde{H}\,\partial_j^2 \tilde{H}\big),
\end{split}
\end{align}
where in the last equality we have used the definition of $\tilde{H}$ in \eqref{eqn:tildeH-definition}.
Plugging \eqref{eqn:trdeltatildeH2} into \eqref{eqn:rewJtilde}, we get that
\begin{equation}
\label{eqn:eqn:rewJtilde2}  
\begin{aligned}
\Tr\big(\Jj^{-+}\Jj^{+-}\big)&=\frac{1}{4}\partial_j^2\big((\Lambda_--\mu)^2 + (\Lambda_+-\mu)^2\big)- \frac{1}{2} {(\partial_j\Lambda_-)}^2 - \frac{1}{2} {(\partial_j\Lambda_+)}^2\\
&\quad-\frac{1}{2} \Tr\big(\tilde{H}\,\partial_j^2 \tilde{H}\big) 
   -\Tr\big(P_e(\partial_j \tilde{H})^2 P_e\big).
   \end{aligned}
\end{equation}
Recalling \eqref{eqn:fjjsing-two-eigenvalues-current}, up to this point we have proven that for any $\eta>0$:
\begin{equation*}
\fjjsing(\eta) = \zjj(\eta) + \xi(\eta),
\end{equation*}
where $\zjj$ is defined in \eqref{eqn:zjj-definition} and  
\begin{equation}
\label{eqn:xi}
\begin{aligned}
    \xi(\eta) &:= \int_{B_\eps} dk \,\Xi(\eta, k),\\
    \Xi(\eta, k)&:=\frac{1}{{(2\pi)}^2} \frac{2(\Lambda_+(k) - \Lambda_-(k))}{\eta^2 + (\Lambda_+(k) - \Lambda_-(k))^2}
     \bigg(\Tr\big(\tilde{H}(k)^2 (\partial_j P_e(k))^2\big) + \frac{1}{2} \Tr\big(\tilde{H}(k)\partial_j^2\tilde{H}(k)\bigg),
    \end{aligned}
\end{equation}
where we have used that $\tilde{H}(k)\partial_j P_e(k)=-\partial_j\tilde{H}(k) P_e(k)$ as $\tilde{H}(k)\, P_e(k)=0$. 
We conclude the proof by showing that the function $\xi$ is differentiable and that $\lim_{\eta\to 0^+} \xi'(\eta)= 0$, that is $\xi$ does not contribute to the longitudinal conductivity $\sigma_{jj}$. 
We shall compute explicitly the derivative $\xi'(\eta)$ for $\eta\to 0$ on each $B_\eps^{(l)}$ (whose disjoint union gives the whole set $B_\eps$, recall \eqref{eqn:definition-Bepsilon}) with $1\leq l\leq n$, by using the  dominated convergence theorem. 
First, for \alev in $k\in B_\eps$ we observe that 
\begin{equation}
\label{eqn:derXi}
\partial_\eta \Xi(\eta, k)=\frac{1}{{(2\pi)}^2} \frac{4\eta(\Lambda_-(k) - \Lambda_+(k))}{{\(\eta^2 + (\Lambda_+(k) - \Lambda_-(k))^2\)}^2}
     \bigg(\Tr\big(\tilde{H}(k)^2 (\partial_j P_e(k))^2\big) + \frac{1}{2} \Tr\big(\tilde{H}(k)\partial_j^2\tilde{H}(k)\bigg).
\end{equation}
Second, we notice that there exists a constant $C_1$:
\begin{equation}
\label{eqn:estxi}
\abs{\Tr\big(\tilde{H}(k)^2 (\partial_j P_e(k))^2\big)}\leq C_1 \abs{k-\omega_l},\qquad\abs{\Tr\big(\tilde{H}(k)\partial_j^2\tilde{H}(k)\big)}\leq C_1\abs{k-\omega_l}\quad\text{\alev $k\in B_\eps^{(l)}$},  
\end{equation}
by using that in view of Assumption \ref{assum:H}\ref{item:conical} there exists a constant $C_2$ such that $\norm{\tilde{H}(k)}\leq C_2\abs{k-\omega_l}$ for \alev $k\in B_\eps^{(l)}$, and by exploiting the $C^2$ regularity of the maps $B_\eps\ni k\mapsto P_e(k)$ and $B_\eps\ni k\mapsto \tilde{H}(k)$, as shown previously. Therefore, by also using condition \eqref{eqn:condition-cone-eps}, we have that there exists a constant $C_3$ such that
\begin{equation}
\label{eqn:ub1}
\abs{\partial_\eta\Xi(\eta,k)}\leq \frac{C_3\eta\abs{k-\omega_l}^2}{\left(\eta^2 + m^2\abs{k-\omega_l}^2\right)^2}\qquad\text{ for \alev $k\in B_\eps^{(l)},$ for every $ 1\leq l\leq n$.}
\end{equation}
For any $a>0$, let us introduce  the function $F_a\colon (0,\infty)\to \R$ such that $F_a(x):=\frac{xa}{{(x^2+a^2)}^2}$. It is easy to check that $\max_{x>0}F_a(x)=\frac{3\sqrt{3}}{16 a^2}$. Therefore, we can estimate the right-hand side of \eqref{eqn:ub1} as:
\begin{equation*}
\frac{C_3\eta\abs{k-\omega_l}^2}{\left(\eta^2 + m^2\abs{k-\omega_l}^2\right)^2}=\frac{C_3 \abs{k-\omega_l}}{m}F_{(m\abs{k-\omega_l)}}(\eta)\leq \frac{C_3 \abs{k-\omega_l}}{m}\frac{3\sqrt{3}}{16 {(m\abs{k-\omega_l})}^2}=\frac{3\sqrt{3}C_3}{16 m^3 \abs{k-\omega_l}}   
\end{equation*}
which is integrable (uniformly in $\eta$). Thus, the dominated convergence theorem and inequality \eqref{eqn:ub1} implies that
\begin{align*}
\abs{\xi'(\eta)}&=\abs{\sum_{l=1}^n\int_{B_\eps^{(l)}}\di k\, \partial_\eta\Xi(\eta,k)}\leq\sum_{l=1}^n\int_{B_\eps^{(l)}}\di k\, \abs{\partial_\eta\Xi(\eta,k)}
\\
&\leq C_3\sum_{l=1}^n\int_{B_\eps^{(l)}}\di k\, 
\frac{\eta\abs{k-\omega_l}^2}{\left(\eta^2 + m^2\abs{k-\omega_l}^2\right)^2}=\frac{n C_3\eta}{2 m^4}\(\ln\(1+\frac{\eps^2 m^2}{\eta^2}\) +\frac{1}{1+\frac{\eps^2 m^2}{\eta^2}}-1\)\to 0\text{ as $\eta\to 0^+$}.  
\end{align*}
\end{proof}



\begin{thebibliography}{00}

\bibitem{AG} M. Aizenman and  G. M. Graf. Localization bounds for an electron gas. {\it J. Phys. A: Math. Gen.} {\bf 31}, 6783 (1998).

\bibitem{avron} J. E. Avron, R. Seiler, and B. Simon. Homotopy and quantization in condensed matter physics. {\it Phys. Rev. Lett.} {\bf 51}, 51 (1983).

\bibitem{AS2} J. E. Avron, R. Seiler, and B. Simon. Charge deficiency, charge transport and comparison of dimensions. {\it Commun. Math. Phys.} {\bf 159}, 399-422 (1994).

\bibitem{ASY} J. E. Avron, R. Seiler, and L. G. Yaffe. Adiabatic theorems and applications to the quantum Hall effect. {\it Comm. Math. Phys.} {\bf 110}, 33-49 (1987).

\bibitem{BdRFrev} S. Bachmann, W. De Roeck, and M. Fraas. The adiabatic theorem in a quantum many-body setting. {\it Contemp. Math.}, {\it Analytic Trends in Mathematical Physics} {\bf 741}, 43-58 (2020).

\bibitem{BdRF} S. Bachmann, W. De Roeck, and M. Fraas. The adiabatic theorem and linear response theory for extended quantum systems. {\it Comm. Math. Phys.} {\bf 361}, 997-1027 (2018).

\bibitem{BBdRF} S. Bachmann, A. Bols, W. De Roeck, and M. Fraas. A Many-Body Index for Quantum Charge Transport. {\it Comm. Math. Phys.} {\bf 375}, 1249-1272, (2020).

\bibitem{BRFL}
S. Bachmann, W. De Roeck, M. Fraas, and M. Lange. Exactness of linear response in the quantum Hall effect. {\it Ann. Henri Poincaré} {\bf 22}, 1113-1132 (2021).


\bibitem{BES} 
J. Bellissard, A. van Elst, and H. Schulz-Baldes. The non-commutative geometry of the quantum Hall effect. {\it J. Math. Phys.} {\bf 35}, 5373 (1994).

\bibitem{CFLSD}
E. Cancès, C. Fermanian Kammerer, A. Levitt, and S. Siraj-Dine. Coherent Electronic Transport in Periodic Crystals. {\it Ann. Henri Poincaré} {\bf 22}, 2643–2690 (2021).


\bibitem{dREF} W. De Roeck, A. Elgart, and M. Fraas. Derivation of Kubo’s formula for disordered systems at zero temperature. {\it Invent. math.} {\bf 235}, 489–568 (2024).

\bibitem{drouot}
A. Drouot. Ubiquity of conical points in topological insulators. {\it  Journal de l’École polytechnique -- Mathématiques} {\bf 8}, 507-532 (2021).

\bibitem{ES}
A. Elgart and B. Schlein. Adiabatic charge transport and the Kubo formula for Landau-type Hamiltonians. {\it Commun. Pure Appl. Math.} {\bf 57}, 590-615 (2004).

\bibitem{Frohlich}
J. Froehlich, U.M. Studer and E. Thiran. Quantum Theory of Large Systems of Non-Relativistic Matter. {\tt arXiv:cond-mat/9508062}

\bibitem{FGR}
S. Fabbri, A. Giuliani and R. Reuvers. Universality of the topological phase transition in the interacting Haldane model. {\it Phys. Rev. B} {\bf 112}, 165113 (2025).

\bibitem{FrMa}
L. Fresta, G. Marcelli. Spin Transport and Lack of Quantisation for Time-Reversal Symmetric Insulators on the Honeycomb Structure. {\it Ann. Henri Poincaré} (2026).

\bibitem{GMJP16}
A. Giuliani, I. Jauslin, V. Mastropietro, and M. Porta. Topological phase transitions and universality in the Haldane-Hubbard model.
{\it Phys. Rev. B} {\bf 94}, 205139 (2016).

\bibitem{GM}
A. Giuliani and V. Mastropietro. The 2D Hubbard model on the honeycomb lattice. {\it Commun. Math. Phys.} {\bf 293}, 301-346 (2010).

\bibitem{GMP12}
A. Giuliani, V. Mastropietro, and M. Porta. Universality of Conductivity in Interacting Graphene. {\it Commun. Math. Phys.} {\bf 311}, 317–355 (2012).

\bibitem{GMPhall} A. Giuliani, V. Mastropietro, and M. Porta. Universality of the Hall Conductivity in Interacting Electron Systems. {\it Comm. Math. Phys.} {\bf 349}, 1107-1161 (2017).

\bibitem{GMPhald} A. Giuliani, V. Mastropietro, and M. Porta. Quantization of the interacting Hall conductivity in the critical regime. {\it J. Stat. Phys.} {\bf 180}, 332-365 (2020).

\bibitem{GMPweyl} A. Giuliani, V. Mastropietro, and M. Porta. Anomaly non-renormalization in interacting Weyl semimetals. {\it Comm. Math. Phys.} {\bf 384}, 997-1060 (2021).

\bibitem{G}
G. M. Graf. Aspects of the integer quantum Hall eﬀect. {\it Proc. Symp. Pure Math.}
{\bf 76}, 429-442 (2007).

\bibitem{GLMP24}
R. L. Greenblatt, M. Lange, G. Marcelli, and M. Porta. Adiabatic Evolution of Low-Temperature Many-Body Systems. {\it Commun. Math. Phys.} {\bf 405}, 75 (2024).

\bibitem{haldane}
F. D. M. Haldane. Model for a Quantum Hall Eﬀect without Landau levels: condensed-matter realization of the \virg{parity anomaly}. {\it Phys. Rev. Lett.} {\bf 61}, 2015--2018 (1988).

\bibitem{rootspolynomial}
G. Harris and C. Martin. Shorter Notes: The Roots of a Polynomial Vary Continuously as a Function of the Coefficients.
{\it Proceedings of the American Mathematical Society}, {\bf 100}, 390-392 (1987).

\bibitem{HM} M. B. Hastings and S. Michalakis. Quantization of Hall conductance for interacting electrons on a torus. {\it Comm. Math. Phys.} {\bf 334}, 433-471 (2015).

\bibitem{henheik}
J. Henheik and S. Teufel. Justifying Kubo’s formula for gapped systems at zero temperature: A brief review and some new result. {\it Reviews in Mathematical Physics}, { \bf 33},  No. 01,  2060004 (2021).

\bibitem{KS}
M. Klein and R. Seiler. Power-law corrections to the Kubo formula vanish in quantum Hall systems.
{\it Commun. Math. Phys.}, {\bf 128}, 141-160 (1990).


\bibitem{klitzing}
K. von Klitzing, G. Dorda, and M. Pepper. New Method for High-Accuracy Determination of the Fine-Structure Constant Based on Quantized Hall Resistance. {\it Phys. Rev. Lett.} {\bf 45}, 494-497 (1980).

\bibitem{kubo}
R.~Kubo. Statistical-mechanical theory of irreversible processes I: General theory and simple applications to magnetic and conduction problems. {\it J. Phys. Soc. Jpn.} {\bf 12}, 570--586 (1957).

\bibitem{laughlin} 
R. B. Laughlin. Quantized Hall conductivity in two dimensions. {\it Phys. Rev. B} {\bf 23}, 5632 (1981).

\bibitem{ludwig}
A. W. W. Ludwig, M. P. A. Fisher, R. Shankar, and G. Grinstein. Integer quantum Hall transition: An alternative approach and exact results. {\it Phys. Rev. B} {\bf 50}, 7526 (1994).

\bibitem{M}
G. Marcelli. Improved energy estimates for a class of time-dependent perturbed Hamiltonians. {\it Lett. Math. Phys.} {\bf 112}, 51 (2022).

\bibitem{MM2}
G. Marcelli and D. Monaco. From charge to spin: Analogies and differences in quantum transport coefficients. {\it J. Math. Phys.} {\bf 63}, 072102 (2022)

\bibitem{MM}
G. Marcelli and D. Monaco. Purely linear response of the quantum Hall current to space-adiabatic perturbations. {\it Lett. Math. Phys.} {\bf 112}, 91 (2022). 

\bibitem{mamomopa}
G. Marcelli, D. Monaco, M. Moscolari, and G. Panati. The Haldane model and its localization dichotomy . {\it Rend. Mat. Appl.} {\bf 39}, 307–327 (2018).

\bibitem{MPTa}
G. Marcelli, G. Panati, and C. Tauber. Spin conductance and spin conductivity in topological insulators:
analysis of Kubo-like terms. {\it Ann. Henri Poincaré} {\bf 20}, 2071–2099 (2019).

\bibitem{MPTe}
G. Marcelli, G. Panati, and S. Teufel. A new approach to transport coefficients in the quantum spin Hall
effect. {\it Ann. Henri Poincaré} {\bf 22}, 1069–1111 (2021).

\bibitem{mh} M. Manna\"i and S. Haddad. Strain tuned topology in the Haldane and the modified Haldane models. {\it J. Phys.: Condens. Matter} {\bf 32}, 225501 (2020).

\bibitem{MT} D. Monaco and S. Teufel. Adiabatic currents for interacting electrons on a lattice. {\it Rev. Math. Phys.} {\bf 31}, 1950009 (2019).

\bibitem{strainrev} G. G. Naumis, S. Barraza-Lopez, M. Oliva-Leyva and H. Terrones. Electronic and optical properties of strained graphene and other strained 2D materials: a review. {\it Rep. Prog. Phys.} {\bf 80} 096501 (2017).

\bibitem{panatisparberteufel}
G.~Panati, C.~Sparber, and S.~Teufel. 
Geometric currents in piezoelectricity. {\it Arch. Rat. Mech. Analysis} {\bf 191}, 387--422 (2009).


\bibitem{PS}
M.~Porta and H.~P.~Singh. Large Scale Response of Gapless $1d$ and Quasi-$1d$ Systems. {\it Ann. Henri Poincaré} (2025). \url{https://doi.org/10.1007/s00023-025-01600-z}

\bibitem{PSS}
M.~Porta, G.~Scola and H.~P.~Singh. Large Scale Dynamical Response of Interacting $1d$ Fermi Systems. {\tt arXiv:2509.08665}



\bibitem{stauber}
T. Stauber, N. M. R. Peres, and A. K. Geim. Optical conductivity of graphene in the visible region of the spectrum. {\it Phys. Rev. B} {\bf 78}, 085432 (2008).

\bibitem{teufel}
S. Teufel. Adiabatic Perturbation Theory in Quantum Dynamics. No. 1821 in Lecture Notes in Mathematics. Springer, Berlin (2003).

\bibitem{tknn}
D. J. Thouless, M. Kohmoto,  M. P. Nightingale, and M. den Nijs. Quantized Hall conductance in a two-dimensional periodic potential. {\it Phys. Rev. Lett.} {\bf 49}, 405 (1982).

\bibitem{WMMMT}
M. Wesle, G. Marcelli, T. Miyao, D. Monaco, and S. Teufel. Near Linearity of the Macroscopic Hall Current Response in Infinitely Extended Gapped Fermion Systems. {\it Commun. Math. Phys.} {\bf 406}, 199 (2025).
\end{thebibliography}
\end{document}